\documentclass[11pt]{article}
\usepackage[dvipsnames,svgnames,x11names,hyperref]{xcolor}
\usepackage{amsthm}
\usepackage{amsmath}
\usepackage{amssymb}
\usepackage[colorlinks]{hyperref}

\usepackage{palatino}
\usepackage{todonotes}
\usepackage{mdframed}
\usepackage{tikz}
\usepackage{mdframed}

\usepackage{enumitem}
\usepackage{longtable}
\usepackage{tabu}
\usepackage{stfloats}
\usepackage{framed}
\usepackage{color,fancybox,graphicx,url,subfigure}
\usepackage{nicefrac}
\hypersetup{
  citecolor={blue}
}

\usepackage{float}
\floatstyle{ruled}
\newfloat{test}{ht}{testfile}
\floatname{test}{Test}

\def\showauthornotes{1}

\def\showdraftbox{1}

\newtheorem{theorem}{Theorem}[section]

\newtheorem{definition}[theorem]{Definition}
\newtheorem{lemma}[theorem]{Lemma}
\newtheorem{remark}[theorem]{Remark}

\newtheorem{corollary}[theorem]{Corollary}
\newtheorem{claim}[theorem]{Claim}
\newtheorem{fact}[theorem]{Fact}

\def\FullBox{\hbox{\vrule width 6pt height 6pt depth 0pt}}

\def\qed{\ifmmode\qquad\FullBox\else{\unskip\nobreak\hfil
\penalty50\hskip1em\null\nobreak\hfil\FullBox
\parfillskip=0pt\finalhyphendemerits=0\endgraf}\fi}

\def\qedsketch{\ifmmode\Box\else{\unskip\nobreak\hfil
\penalty50\hskip1em\null\nobreak\hfil$\Box$
\parfillskip=0pt\finalhyphendemerits=0\endgraf}\fi}

\renewenvironment{proof}{\begin{trivlist} \item {\bf Proof:~~}}
   {\qed\end{trivlist}}

\newcommand\R{\mathbb R}

\newcommand{\marginlabel}[1]%
{\mbox{}\marginpar{\it{\raggedleft\hspace{0pt}#1}}}
\newcommand\card[1]{\left| #1 \right|} 
\newcommand\poly{\mbox{poly}}  

\definecolor{Mygray}{gray}{0.8}

 \ifcsname ifcommentflag\endcsname\else
  \expandafter\let\csname ifcommentflag\expandafter\endcsname
                  \csname iffalse\endcsname
\fi

\ifnum\showauthornotes=1
\else
\fi

\ifnum\showauthornotes=1
\newcommand{\Authornote}[2]{{\sf\small\color{red}{[#1: #2]}}}
\newcommand{\Authoredit}[2]{{\sf\small\color{red}{[#1]}\color{blue}{#2}}}
\newcommand{\Authorcomment}[2]{{\sf \small\color{Mygray}{[#1: #2]}}}
\newcommand{\Authorfnote}[2]{\footnote{\color{red}{#1: #2}}}
\newcommand{\Authorfixme}[1]{\Authornote{#1}{\textbf{??}}}
\newcommand{\Authormarginmark}[1]{\marginpar{\textcolor{red}{\fbox{
#1:!}}}}
\else
\newcommand{\Authornote}[2]{}
\newcommand{\Authoredit}[2]{}
\newcommand{\Authorcomment}[2]{}
\newcommand{\Authorfnote}[2]{}
\newcommand{\Authorfixme}[1]{}
\newcommand{\Authormarginmark}[1]{}
\fi

\newcommand{\inbraces}[1]{\left\{#1\right\}}           

\def\implies{\Rightarrow}

\renewcommand\Pr{\mathop{\mbox{\bf Pr}}}
\newcommand\av{\mathop{\mbox{\bf E}}}

\newcommand{\Ex}[2]{\av_{{#1}}\left[{#2}\right]}

\def\abs#1{\left| #1 \right|}
\newcommand{\norm}[1]{\ensuremath{\left\lVert #1 \right\rVert}}

\newlength{\pgmtab}  
 {
	\begin{enumerate}}{\end{enumerate}}

\newcounter{lecnum}

\newlength{\tpush}
\ifnum\showdraftbox=1

\else

\fi

\newcommand{\cons}{\mathcal{C}}
\newcommand{\F}{\mathbb{F}}
\newcommand{\E}{\mathop{\mbox{\bf E}}}
\newcommand{\st}{\textbf{ s.t. }}

\newcommand{\fail}{\bot}
\newcommand{\one}{\mathbf{1}}
\newcommand{\approxparam}[1]{\,{\stackrel{{#1}}{\approx}}\,}

\newcommand{\eps}{\epsilon}

\def\D{{\cal D}}

\newcommand{\ls}{{ S}}
\newcommand{\ks}{{ K}}
\newcommand{\rs}{{ R}}
\newcommand{\LL}{s}
\newcommand{\KL}{k}
\newcommand{\RL}{r}

\newcounter{listcounter}

\hypersetup{linkcolor=blue,citecolor=blue,filecolor=blue,urlcolor=blue}

\newcommand{\lref}[2][]{\hyperref[#2]{#1~\ref*{#2}}}
\renewcommand{\eqref}[1]{\hyperref[#1]{(\ref*{#1})}}

\title{Cube vs. Cube Low Degree Test}

\author{Amey Bhangale
\thanks{Department of Computer Science, Rutgers University, USA.  {\tt amey.bhangale@rutgers.edu}}
\and
Irit Dinur\thanks{Faculty of Computer Science and Mathematics, Weizmann Institute, Rehovot, Israel.
{\tt irit.dinur@weizmann.ac.il}}
\and
Inbal Rachel Livni Navon\thanks{Computer Science and Mathematics, Weizmann Institute, Rehovot, Israel. {\tt inbal.livni@weizmann.ac.il}}
}

\begin{document}

\maketitle

\begin{abstract}

We revisit the Raz-Safra plane-vs.-plane test and study the closely related cube vs. cube test. In this test the tester has access to a ``cubes table'' which assigns to every cube a low degree polynomial. The tester randomly selects two cubes (affine sub-spaces of dimension $3$) that intersect on a point $x\in \F^m$, and checks that the assignments to the cubes agree with each other on the point $x$.
Our main result is a new combinatorial proof for a low degree test that comes closer to the soundness limit, as it works for all  $\eps \ge \poly(d)/{\card\F}^{1/2}$, where $d$ is the degree. This should be compared to the previously best soundness value of $\eps \ge \poly(m, d)/\card\F^{1/8}$. Our soundness limit improves upon the dependence on the field size and does not depend on the dimension of the ambient space.

Our proof is combinatorial and direct: unlike the Raz-Safra proof, it proceeds in one shot and does not require induction on the dimension of the ambient space.
The ideas in our proof come from works on direct product testing which are even simpler in the current setting thanks to the low degree.

Along the way we also prove a somewhat surprising fact about connection between different agreement tests:
it does not matter if the tester chooses the cubes to intersect on points or on lines: for every given table, its success probability in either test is nearly the same.

\end{abstract}

\def\acxc{\alpha_{{\cal C}x{\cal C}}}
\def\aplp{\alpha_{{\cal P}\ell {\cal P}}}

\def\aclc{\alpha_{{\cal C}\ell{\cal C}}}
\def\apxp{\alpha_{{\cal P}x {\cal P}}}

\section{Introduction}

Low degree tests are local tests for the property of being a low degree function. These were the first property testing results that were discovered, and are an important component in PCP constructions. Such tests were studied in the 1990's and their ballpark soundness behavior was more or less understood. In this work we revisit these tests and give a new and arguably simpler analysis for the cube vs. cube low degree test. Our proof method allows us to get a soundness guarantee that is much closer to the conjectured optimal value. Discovering the precise point in which soundness starts to hold is an intriguing open question that captures an interesting aspect of local-testing in the small soundness regime.

Let us begin with a short introduction to low degree tests. A low degree test can be described as a game between a prover and a verifier, in which the prover wants to convince the verifier that a function $f:\F^m\to\F$ is a low degree polynomial. The most straightforward way for the prover to specify $f$ would be to give its value on each point $x\in \F^m$. However, in this way, to check that $f$ has degree at most $d$ the verifier would have to read $f$ on at least $d+2$ points. If we want a verifier that makes fewer queries while keeping the error small, it is useful to move to a more redundant representation of $f$. For example, the verifier can ask the prover to specify for every cube (affine subspace of dimension $3$) $C\subset \F^m$, a function $f_C:C\to\F$ that is defined on the cube and is obtained by restricting $f$ to that cube. This is called a ``cubes-table'', and similarly one can consider a lines table (with an entry for every line), or a planes table (with an entry for each plane).

Thus, in the cubes representation of a low degree function $f:\F^m\to\F$, we have a table entry $T(C)$ for every cube $C$ and the value of that entry is supposed to be $T(C) = f|_C$.  A general cubes table is a table $T(\cdot )$ indexed by all possible cubes and the $C$-th entry is a low degree function on the cube $C$. Each $T(C)$ is viewed as a local function. Indeed the number of bits needed to specify $T(C)$ is only $O(d^3\log \card\F)$ which is much smaller than $\binom{m+d}d\log\card\F$ - the number of bits needed to represent a general degree $d$ function $f$ on $\F^m$.

The prover may cheat, as provers do, by giving a cubes table whose entries cannot be ``glued together'' into any one global low degree function. This is where the {\em agreement} test comes in. The verifier can check the table by reading two entries corresponding to two cubes that have a non-trivial intersection, and checking that the function $T({C_1})$ and the function $T({C_2})$ agree on points in the intersection of $C_1\cap C_2$.

\begin{test}
	\caption{Cube vs. Cube agreement test. \label{test:CxC}}
	\begin{enumerate}
		\item Select a point $x\in \mathbb{F}^m$.
		\item Pick affine cubes $C_1, C_2$ randomly conditioned on $C_1,C_2 \ni x$.
		\item Read $T({C_1}),T({C_2})$ from the table and accept iff $T({C_1})(x) = T({C_2})(x)$.
	\end{enumerate}
Let $\acxc(T)$ be the {\em agreement} of the table $T$, i.e. the probability of acceptance of the test.
\end{test}

The test is local in that it accesses only two cubes.
Different tests may differ in the distribution underlying the agreement test (for example, Raz and Safra look at two planes that intersect in a line, which clearly is a different distribution from choosing two planes that intersect in a point), but they all check agreement on the intersection, so we generally refer to all of these as agreement tests.

The interesting point, as proven by both Raz and Safra in \cite{RaSa}, and by Arora and Sudan in \cite{ArSu}, is that such tests have small soundness error. For example, the plane vs. plane theorem of Raz Safra is as follows,

\begin{theorem}[Raz-Safra \cite{RaSa}]
There is some $\delta>0$ such that for every $d$ and prime power $q$ and every $m\ge 3$ the following holds.
Let $\F$ be a finite field $\card\F=q$, and let $T(\cdot)$ be a planes table, assigning to each plane $P\subset \F^m$ a bivariate degree $d$ polynomial $T(P):P\to\F$. Let $\aplp(T)$ be as defined in \lref[Test]{test:PlP}.

For every $\eps \ge (md/q)^\delta$, if $\alpha_{P\ell P} (T) \ge \eps$ then there is a degree $d$ function $g:\F^m\to\F$ such that $T(P) = g|_P$ on an $\Omega(\eps)$ fraction of the planes.
\end{theorem}

\begin{test}
	\caption{The Raz-Safra Plane vs. Plane agreement test. \label{test:PlP}}
	\begin{enumerate}
		\item Select an affine line $\ell\subset \mathbb{F}^m$.
		\item Choose affine planes $P_1, P_2$ randomly conditioned on $P_1,P_2 \supset\ell$.
		\item Read $T({P_1}),T({P_2})$ from the table and accept iff $T({P_1})(x) = T({P_2})(x)$ for all $x\in \ell$.
	\end{enumerate}
Let $\aplp(T)$ be the {\em agreement} of the table $T$, i.e. the probability of acceptance of the test.
\end{test}


A similar theorem was proven by Arora and Sudan for $T$ a lines table and for a natural test that checks if two intersecting lines agree on the point of intersection.

These results are called low degree tests although it makes sense to think of them as theorems relating local agreement to global agreement. We refer to them as low degree {\em agreement} test theorems.
\paragraph{Towards the soundness threshold.} The most important aspect of the low degree agreement theorems of \cite{RaSa,ArSu} is the fact that they have small soundness. Small soundness means that a cheating prover won't be able to fool the verifier into accepting with even a tiny $\eps>0$ probability, unless the table has some non-trivial agreement with a global low degree function. Small soundness of low degree tests was used inside PCP constructions for getting PCPs with the smallest known soundness error. The fact that soundness holds for all values of $\eps \ge (d/q)^{\delta}$ was sufficient for the PCP constructions of \cite{RaSa,ArSu}. It is likely that finding the minimal threshold beyond which soundness is guaranteed to hold will be important for determining the best possible PCP gaps.

Regardless of the PCP application, this encoding of a function $f$ by its restrictions to cubes (or to planes) is quite natural, and is a rare example of a property that has such strong testability. The low degree agreement test theorems guarantee that even the passing of the test with tiny $\eps$ probability has non-trivial structural consequences. Perhaps the best known comparable scenario is that of the long code, defined in \cite{BGS}, that has similar properties, and for which an extensive line of work has been able to determine the precise threshold of soundness. Another setting with a similarly strong soundness is related to the inverse theorems for the Gowers uniformity norms. In that setting the function is given as a points-table, and the Gowers norm measures success in a low degree test, so it is not altogether dissimilar from the situation here.


To summarize, one of our goals is to pinpoint the absolute minimal soundness value for which a theorem as above holds. Can this threshold be, as it is in the aforementioned cases, as small as the value of a random assignment?
In other words, could it be true that for every table whose agreement parameter is an additive $\eps>0$ above the value that we expect from a random table, already some structure exists?

The best known value for $\delta$ for the plane vs. plane test is due to Moshkovitz and Raz who proved in \cite{MR-LDT} that the plane vs. plane test has soundness for all $\eps \ge \poly(d)/q^{1/8}$. But what is the correct exponent of $q$ ?

We make progress on this question not for the plane vs. plane test but rather for the cube vs. cube test. For our test, since the intersection consists of one point, the soundness can not go below $1/q$ because the agreement of {\em every} table, even a random one, is always at least $1/q$.

Our main theorem is,
\begin{theorem}\label{thm:main}
There exist constants $\beta_1,\beta_2 > 0$ such that for every $d$, large enough prime power $q$ and every $m\ge 3$ the following holds:

Let $\F$ be a finite field, $\card\F=q$. Let $T$ be a cubes table, assigning to each cube $C\subset\F^m$ a degree $d$ polynomial $T(C):C\to\F$. Let $\acxc(T)$ be as defined in \lref[Test]{test:CxC}.  If $\acxc(T) \ge \epsilon $ for $\epsilon \geq \beta_1 d^4/q^{1/2} $, then there is a degree $d$ function $g:\F^m\to\F$ such that $T(C) = g|_C$ on an $\beta_2\epsilon$ fraction of the cubes.

\end{theorem}

The improvement over previous theorems is that the dependence on $q$ is $1/q^{1/2}$ compared to $1/q^{1/8}$,
It is an intriguing question whether the dependence on $q$ can be made inversely linear, i.e. $1/q$.


\begin{remark}
We don't know the precise dependence of $\epsilon$ on the degree $d$. In this work we made no attempt to optimize this dependence. We would like to point out that our proof can be modified to change the dependence from $d^4$ to $d^3$. See \lref[Remark]{remark:improve d} for more details.
\end{remark}

\paragraph{Simplified analysis.} While the line vs. line test considered by Arora and Sudan \cite{ArSu} is the most natural to come up with, it is rather difficult to analyze. In contrast, one of the captivating aspects of the Raz-Safra proof is that it is combinatorial, and the low degree aspect of the table plays a role only in that it guarantees distance between distinct polynomials on a line.
Our analysis continues this combinatorial approach, and further simplifies it. Unlike the Raz-Safra proof, we do not need to use induction on the dimension of the ambient space $m$ but rather recover the global structure from $T$ ``in one shot''. We rely on ideas from direct product testing, \cite{DG,IKW,DS}, and on some spectral properties of incidence graphs such as the cube-point graph.

\paragraph{Proof Outline.}
Given a table $T$, whose agreement is some small $\eps$, the proof must somehow come up with the global low degree function $g:\F^m\to\F$ and then argue that on many of the cubes indeed $T(C) = g|_C$. Naively, we might try to define $g$ at each point $x$ according to the most common value among all cubes containing $x$.
This is a viable approach when the agreement is close to $1$, as is done, e.g. in the linearity testing theorem of \cite{BLR}. However, when the agreement is a small $\eps>0$, this will simply not work as we can see by considering the table half of whose entries are $T(C) \equiv 0$ and the other half $T(C) \equiv 1$. The agreement of this table is an impressive $\acxc(T) = 1/2$, and yet the suggested definition of $g$ according to majority will yield a random function that might be quite far from any low degree function.

We get around this problem by taking a {\em conditional majority}. For every point $x\in \F^m$ and value $\sigma\in \F$ we consider only cubes containing $x$ for which $T(C)(x)=\sigma$. These cubes already agree with each other on $x$ and are thus likely to agree on any other point of their intersection. Since the cubes containing $x$ cover every $y\in\F^m$, we can define a function $f_{x,\sigma}:\F^m\to\F$ on the entire space $\F^m$ by taking the most popular value among these cubes (i.e. the set of cubes whose value on $x$ is $\sigma$). We choose a best $\sigma$ for each $x$ and are left with a global function $f_x$ for each $x$.

The proof proceeds in three steps.
\begin{itemize}
\item Local structure: We show that this conditional majority definition is good, obtaining for each $x$ and $\sigma$ a function $f_{x}:\F^m\to\F$ that is ``local'' in that it comes from the cubes containing a point $x$. This is done in \lref[Section]{sec:local}.
\item Global Structure: We then show that there are many pairs $x,y$ for which $f_{x} \approx f_{y}$ thus finding a global $g$ that agrees with many of the cubes. This is done in \lref[Section]{sec:global}.
\item Low Degree: Finally, we show that $g$ is very close to a true low degree function. This is done by reduction to the Rubinfeld-Sudan low degree test \cite{RS96} that works in the high-soundness regime. This is done in \lref[Section]{sec:lowdeg}.
\end{itemize}

\paragraph{Agreement tests: low degree tests and direct product tests.}
The proof outline above resembles works on direct product testing, and this is no coincidence.
The low degree testing setting can be generalized to a more abstract ``agreement testing'' in which a function $f:X\to\Sigma$ is represented not as a truth table but as a collection of restrictions $(f|_S)_{S \in {\cal S}}$ where ${\cal S} = \inbraces{S\subset X}$ is a collection of subsets of $X$. A natural agreement test can be defined and studied. This type of question was first suggested in work of Goldreich and Safra \cite{GolSaf} in an attempt to separate the algebraic aspect of the low degree test from the combinatorial. There has been a follow-up line of work on this, \cite{DR,DG,IKW,DS}, focusing especially on the case where $X$ is a finite set, $X = [n]$, and ${\cal S}$ is the collection of all $k$-element subsets of $X$.

In the work here we bring some of the ideas from that line of work, most notably from \cite{IKW}, back to the low degree testing question. The fact that our table entries have low degree gives us extra power which makes our proof simpler than that in the abstract setting, yielding a particularly direct proof of a low degree agreement test.

Our proof makes an explicit use of the expansion properties of the relevant incidence graphs (cube vs. line, cube vs. point etc.). This allows us to prove that for every table $T$, different tests have similar agreement.
\begin{lemma}\label{lem:pxp plp}
Let $T$ be a planes table, and let $\apxp(T)$ be the success probability of a test with two planes that intersects on a point. Let $ \aplp(T)$ be the success probability of  \lref[Test]{test:PlP}, then
\[ \apxp(T)\left( - \frac{d}{q}\right)\leq \aplp(T) \leq \apxp(T) + \frac{1}{q}(1+o(1)). \]
\end{lemma}
In fact, we proved a more general equivalence between tests, the general statement appears on \lref[Section]{sec:tests}.

\section{Preliminaries and Notations}
\subsection{Notations}
All the graphs we discuss throughout the paper are bipartite bi-regular graphs. Given such graph $G$, whose sides are $A,B$ we denote by $\one$ the all one vector, its size will be implied by the context. For a subset of vertices $A'\subset A$, we denote by $\one_{A'}$ the indicator vector for $A'$. For a vertex $a\in A$, we denote by $N(a)\subseteq B$ the neighbors of $a$ in $G$.

We use normalized inner product, such that for $x,y\in\R^{n}$, $\langle x,y\rangle = \frac{1}{n}\sum_i x_i y_i$, which means that $\langle \one,\one\rangle =  1$. The norm is defined by $\norm{x} = \sqrt{\langle x,x\rangle}$.

We use the notation $x\sim S$ to denote $x$ being sampled uniformly at random (u.a.r) from the set $S$, in case this set $S$ equals the entire space, we omit this symbol and simply write $\Pr_a$ or $\E_a$ to describe choosing a uniform vertex $a\in A$. We use the notation $\mathbb{I}(E)$ to denote the indicator random variable of the event $E$.

For two vectors $u,v$, we use the notation $u\approxparam{\gamma}v$ if $u$ and $v$ are equal on at least $1-\gamma$ of the coordinates.

Fix a vector space $\F^m$. An affine space of $S$ dimension $k$ is defined by $k+1$ vectors $x_0,x_1,\ldots,x_k$ such that $x_1,\ldots,x_k$ are linearly independent,
\[ S = x_0 + \mathtt{span}(x_1,\ldots,x_k) = \{ x_0 + t_1 x_1 + \ldots t_k x_k \;| \; t_1,\ldots,t_k \in \F. \}\]
A {\em line} is a $1$-dimensional affine space, a {\em plane} is a $2$-dimensional affine space, and a {\em cube} is a $3$-dimensional affine space.
We will denote the set of all lines and cubes by $\mathcal{L}$ and $\mathcal{C}$ be respectively. For a point $x\in \F^m$ let
\[\mathcal{L}_x = \{ \ell\in\mathcal{L}\;|\; \ell\ni x\}\qquad \mathcal{C}_x = \{ C\in\mathcal{C}\;|\; C \ni x\}.\]
Similarly for a line $\ell\in\mathcal{L}$ let $\mathcal{C}_\ell$ be the set of all cubes that contains $\ell$.

\subsection{Spectral Expansion Properties}
In this section, we prove two properties of bi-regular bipartite graphs with good spectral parameters. In an expander, the following is well known: if we sample a random neighbor of a small, but not too small, set of vertices, we get a nearly uniform distribution over the entire set of vertices. For our purposes, we will require something more. We need to consider not only the distribution over the vertices, but also the distribution over the edges. This is done in two lemmas below.

\begin{definition}\label{def:lambda}
Let $G = (A\cup B,E)$ be a bi-regular bipartite graph, and let $M\in\mathbb{R}^{A\times B}$ be the adjacency matrix normalized such that $\norm{M\one} = 1$, denote by $\lambda(G)$ the value
\[ \lambda(G) = \max_{v\perp \one}\left\{ \frac{\norm{Mv}}{\norm{v}} \right\} .\]
\end{definition}
This is really the second largest singular value of $M$, with a different normalization (such that the maximal singular value equals 1).

\begin{definition}
\label{def:D1D2}
Let $G = (A\cup B,E)$ be a bi-regular bipartite graph and let $B'\subseteq B$ be a subset of vertices. Define the following two distributions $D_i:A\times B\cup \fail\rightarrow[0,1]$ for $i=1,2$.
\begin{itemize}
\item  $D_1$ : Pick $b\in B'$ u.a.r. then pick $a\in N(b)$ u.a.r.
\item  $D_2$ : Pick $a\in A$ u.a.r. If $B'\cap N(a) = \emptyset$, return $\fail$. Else, pick $b\in N(a)\cap B'$ u.a.r.
\end{itemize}
\end{definition}
Clearly if $B'=B$ then $D_1=D_2$. Moreover, if $G$ is sufficiently expanding, then even for smaller $B' \subsetneq B$, the distributions are similar. Indeed, for any event defined on the edges, i.e. a subset $E'\subset E$, the following lemma shows that the probability of $E'$ is roughly the same under the two distributions.
\begin{lemma} \label{lemma:edge sampling}
Let $D_1, D_2$ as defined in \lref[Definition]{def:D1D2}. Let $G = (A\cup B,E)$ be a bi-regular bipartite graph, then for every subset $B'\subset B$ of measure $\mu>0$ and every $E'\subset E$
\[ \abs{ \Pr_{(a,b)\sim D_1}[(a,b)\in E'] - \Pr_{(a,b)\sim D_2}[(a,b)\in E']} \leq \frac{\lambda(G)}{\sqrt{\mu}} .\]
Where it is understood that if $D_2$ output $\fail$, we treat it as if $(a,b)\not\in E'$.
\end{lemma}

We now state a similar lemma, for sampling two adjacent edges instead of a single edge. We will need the graph to satisfy one more requirement.
\begin{definition}
\label{def:D3D4}
Let $G=(A\cup B,E)$ be a bi-regular bipartite graph, such that every two distinct  $b_1,b_2\in B$ have exactly the same number of common neighbors (i.e for all distinct $b_1,b_2\in B$, $|N(b_1)\cap N(b_2)|$ is the same), and this number is non-zero.
Let $B'\subseteq B$ be a subset of vertices, we define the following distributions $D_i:(A\times B\times B)\cup \fail\rightarrow[0,1]$, for $i=3,4$.
\begin{itemize}
\item $D_3$ : Pick $b_1,b_2\in B'$ u.a.r. then pick $a\in N(b_1)\cap N(b_2)$ u.a.r.
\item $D_4$ : Pick $a\in A$ u.a.r. If $B'\cap N(a) = \emptyset$, return $\fail$. Else, pick $b_1,b_2\in N(a)\cap B'$ u.a.r.
\end{itemize}
\end{definition}

\begin{lemma}\label{lemma:two edges sampling}
	Let $D_3, D_4$ be as defined in \lref[Definition]{def:D3D4}. Let $G=(A\cup B,E)$ be a bi-regular bipartite graph, such that every two distinct  $b_1,b_2\in B$ have exactly the same number of common neighbors (i.e for all distinct $b_1,b_2\in B$, $|N(b_1)\cap N(b_2)|$ is the same), and this number is non-zero. Then for every subset $B'\subset B$ of measure $\mu>0$ and every $E'\subset E$
	\[ \abs{\Pr_{a,b_1,b_2\sim D_3}[(a,b_1)(a,b_2)\in E'] - \Pr_{a,b_1,b_2\sim D_4}[(a,b_1)(a,b_2)\in E']} \leq \frac{2\lambda(G)}{\mu} +\frac{1}{\mu^2d_A} + \frac{1}{\mu^2\abs{B}}, \]
where $d_A$ is the degree on $A$ side, and it is understood that if $D_4$ output $\fail$, we treat it as if $(a,b)\not\in E'$.	\end{lemma}
The proofs of these two lemmas appear in \lref[Appendix]{section:missing proofs}.

\subsection{Inclusion Graphs and Their Spectral Gap}

We record here the expansion of several bi-partite {\em inclusion graphs} that will be relevant for our analysis. We prove the claims about these spectral gaps in \lref[Appendix]{section:inclusion proofs}. Unless otherwise stated, $G(A,B)$ denotes a bipartite inclusion graph between $A$ and $B$ where $a\in A$ is connected to $b\in B$ if $a\subseteq b$. The relation of containment will be clear from the sets $A$ and $B$.

For example, the in the graph $G_1(\mathcal{L}\setminus\mathcal{L}_x,\mathcal{C}_x)$, the left side vertices $A$ are all the lines that do not contain $x\in\F^m$, and the right side vertices are all the cubes that contain $x$. There is an edge between a line $\ell$ and a cube $C$ if $\ell\subset C$.

Recall \lref[Definition]{def:lambda} of $\lambda(G)$ for a bipartite graph $G$.

\newenvironment{thmenum}
 {\begin{enumerate}[label=\upshape(\arabic*),ref=\theclaim (\arabic*)]}
 {\end{enumerate}}

\begin{lemma}
\label{lemma:ev_bounds}
We have for every $m\geq 6$,
\begin{thmenum}
\item \label{claim:ev_LCx} For $G_1(\mathcal{L}\setminus \mathcal{L}_x, \mathcal{C}_x)$ , $\lambda(G_1) \approx \frac{1}{\sqrt{q}} $.
\item \label{claim:ev_LxCx} For $G_2(\mathcal{L}_x,\mathcal{C}_x))$ , $\lambda(G_2) \approx  \frac{1}{q}$.
\item \label{claim:ev_FmC-ell}For $G_3(\F^m\setminus \ell,\mathcal{C}_\ell)$ , $\lambda(G_3)\approx  \frac{1}{\sqrt{q}} $.
\item \label{claim:ev_FmC} For $G_4(\F^m,\mathcal{C})$ , $\lambda(G_4) \approx  \frac{1}{q^{3/2}}  $.
\item \label{claim:ev_FmCx} For $G_5(\F^m\setminus\{x\},\mathcal{C}_x)$ , $\lambda(G_5) \approx  \frac{1}{q}  $.
\setcounter{listcounter}{\value{enumi}} 
\end{thmenum}
And for every $m\geq 3$
\begin{thmenum}
\setcounter{enumi}{\value{listcounter}}
\item \label{claim:ev_FmL} For $G_6(\F^m,\mathcal{L})$, $\lambda(G_6)\approx \frac{1}{\sqrt{q}}$.
\end{thmenum}
where $\approx$ denotes equality up to a multiplicative factor of $1\pm o(1)$, and $o(1)$ denotes a function that approaches zero as $q\to\infty$.
\end{lemma}
In general one can see that $\lambda \approx \frac 1 {\sqrt q^p}$ where $p$ is the number of degrees of freedom left after choosing a left hand vertex.
We prove this lemma in \lref[Appendix]{section:inclusion proofs}.

\section{Proof of the Main Theorem}
In this section we prove \lref[Theorem]{thm:main} in three steps - local structure, global structure and finally proving the agreement with a low degree polynomial. These parts are proved in the subsequent subsections.

Let $T$ be a degree $d$ cubes table, i.e. for every $C\in \mathcal{C}$, $T(C):C\to\F$ is a degree $d$ polynomial. Further assume that $\acxc(T) \geq \epsilon$, where $\epsilon = \Omega (d^4/\sqrt q)$.

\subsection{Local Structure}\label{sec:local}

In this section we show that for many points $x\in \F^m$, there exists a function $f_x:\F^m\to\F$ for which $f_x|_C \approxparam{2\gamma} T(C)$ for a good fraction of the cubes containing $x$, for $\gamma=\Omega(1/d^3)$. Recall that $\approxparam{2\gamma}$ means that the two functions agree on $1-2\gamma$ fraction of the points in their domain.

For each $x\in\F^m$ and $\sigma\in\F$, we define \[\cons_{x,\sigma} = \{C\in \mathcal{C}_x | T(C)(x) = \sigma \}.\]
Following \cite{IKW} we have the following important definition,
\begin{definition}[Excellent pair]
$(x,\sigma)$ is $(\frac{\epsilon}{2}, \gamma)$-excellent if:
\begin{enumerate}
	\item $\Pr_{C\in\mathcal{C}_x}[C\in\cons_{x,\sigma}]\geq \frac{\epsilon}{2}$.
	\item Let $C_1,\ell,C_2$ be chosen by the following probability distribution, $C_1\in \mathcal{C}_{x,\sigma}$ u.a.r, $\ell\subset C_1$ a random line that contains $x$ and $C_2\in \mathcal{C}_{x,\sigma}\cap \mathcal{C}_\ell$ (a random cube in $\mathcal{C}_{x,\sigma}$ that contains $\ell$). \[\Pr_{C_1,\ell,C_2}[T(C_1)_{|\ell} \neq T(C_2)_{|\ell}]\leq \gamma.\]
\end{enumerate}
A point $x\in\F^m$ is $(\frac{\epsilon}{2}, \gamma)$-excellent, if exists $\sigma\in\F$ such that $(x,\sigma)$ is $(\frac{\epsilon}{2}, \gamma)$-excellent.
\end{definition}
Note that in the definition of excellent, the marginal distribution of both $C_1,C_2$ is uniform in $\cons_{x,\sigma}$.

In the sequel, we fix $\gamma= \Omega(1/d^3)$ and say that a point is excellent if it is $(\frac\eps 2,\gamma)$-excellent.
We now state the main lemma in this section.
\begin{lemma}[Local Structure] \label{lemma:local structure}
For $\gamma = \Omega(\frac{1}{d^3})$, let $T$ be a cubes table that passes \lref[Test]{test:CxC} with probability larger than $\epsilon = \Omega( \frac{d^4}{\sqrt{q}})$,
then at least $\frac{\epsilon}{3}$ of the points $x\in \F^m$ are excellent, and for each excellent $x$ there exist a function $f_x:\F^m\rightarrow \F$ such that
\[ \Pr_{C\sim \mathcal{C}_x}[T(C) \approxparam{2\gamma} f_{x|_C}] \geq \frac{\epsilon}{4} .\]
\end{lemma}

We will consider the distribution $\D$ on $(x,\ell,C_1,C_2)$ obtained by choosing $x$ uniformly, choosing $\ell \in \mathcal{L}_x$ uniformly, and then choosing $C_1,C_2\in \mathcal{C}_\ell$ uniformly.

This distribution induces a distribution $(x,T(C_1)(x))$ on pairs of point $x$ and value $\sigma\in \F$.

\begin{claim}
	\label{claim:good but not excellent}
	For every $\gamma = \Omega(\frac{1}{d^3})$,
	$$ \Pr_{(x,\sigma)}[\mbox{$(x,\sigma)$ is $(\frac{\epsilon}{2}, \gamma)$ - excellent}] \geq \frac{\epsilon}{3}.$$
\end{claim}
\begin{proof}
	We consider $(x,\ell,C_1,C_2)$ chosen according to $\D$, and we note that the marginal distribution over all elements is uniform. We also write $\sigma = T(C_1)(x)$.
	We define the following events on $(x,\ell,C_1,C_2)$:
	\begin{enumerate}
		\item $E$ : ``$\ell$ is confusing for $x$'': $T(C_1)(x) = T(C_2)(x)$, $T(C_1)_{|\ell}\neq T(C_2)_{|\ell}$.
		\item $H$ : ``$x,C_1$ is heavy'': $\Pr_{C\sim\mathcal{C}_x}[T(C)(x) = T(C_1)(x)] \geq \frac{\epsilon}{2}$
	\end{enumerate}
	Since $T(C_1)_{|\ell},T(C_2)_{|\ell}$ are two degree $d$ polynomials, and $x$ is a random point in $\ell$,
	\[ \Pr_{(x,\ell,C_1,C_2)}[E]\leq \frac{d}{q} .\]
	Using the fact that $\acxc(T) \ge \eps$, and averaging, we get
	\begin{equation}\label{eq:1} \Pr_{(x,\ell,C_1,C_2)}[H]\geq \frac{\epsilon}{2} .\end{equation}

	Instead of picking $C_1$ as a uniform cube containing $x$, we can choose it by the following process, pick $\sigma$ proportional to its weight in $\mathcal{C}_x$, then pick $C_1\sim\cons_{x,\sigma}$. This process describes the same distribution.
	
	Note that after deciding $x,\sigma$, the event $H$ is already determined, so \eqref{eq:1} becomes $\Pr_{x,\sigma}[H] \ge \eps/2$.
	Also, notice that conditioned on $x,\sigma$, the distribution $\D$ is choosing $C_1$ uniformly from $\cons_{x,\sigma}$ and then $\ell\subset C_1$ a random line containing $x$ and then $C_2$ a random cube containing $\ell$ (and we do not require that $T(C_2)(x)=\sigma$).
	The event $H$ is already fixed by $x,\sigma$, but the event $E$ will occur only if $C_2 \in \cons_{x,\sigma}$ and also $T(C_1)|_\ell \neq T(C_2)|_\ell$.
	
	We want to bound the probability of $x,\sigma$ such that $H = 1$, but $\E_{C_1,\ell,C_2}[E|x,\sigma] \leq \gamma\cdot\frac{\epsilon}{2}$.
	We know that
	\[ \E_{x,\sigma}\left[ \Pr[H \wedge E  \;|\; x,\sigma] \right] = \Pr [H\wedge E] \leq \Pr [ E]\leq \frac{d}{q} .\]
	Therefore, by averaging, the probability over $x,\sigma$ that we have $\Pr[H\wedge E | x,\sigma] > \eps\gamma/2$ is at most $\frac{d/q}{\eps\gamma/2}$.
	So for at least $\eps/2 - \frac{d/q}{\eps\gamma/2}\ge \eps/3$ of the pairs $x,\sigma$, we have that both $H$ occurs, and that $\E_{C_1,\ell,C_2}[E|x,\sigma] \leq \eps\gamma/2$.
	
	We end by showing that such $x,\sigma$ are excellent. The first requirement follows by the fact that $H$ occurs, for the second we need to show that for $C_1\in\cons_{x,\sigma}$, a uniform $\ell\in\mathcal{C}_1$ and a uniform $C_2\in\cons_{x,\sigma}\cap C_\ell$ the probability of $T(C_1)_{|_\ell}\neq T(C_2)_{|_\ell}$ is lower than $\gamma$.
	
	We notice that after fixing $(x,\sigma)$, the distribution $\D$ chooses $C_1\in\cons_{x,\sigma}$, a uniform $\ell\in\mathcal{C}_1$, but then a uniform $C_2\in\mathcal{C}_\ell$.
	
	The event $E$ can be written as $E = E_1\wedge E_2$ where $E_1$ is the event ``$T(C_1)(x) = T(C_2)(x)$'' and $E_2$ is the event ``$T(C_1)_{|\ell}\neq T(C_2)_{|\ell}$''. In this notation
	\begin{align*}
	\E_{C_1,\ell,C_2}[E|x,\sigma] =& \E_{C_1,\ell,C_2}[E_1\wedge E_2|x,\sigma] \\=&
	\E_{C_1,\ell,C_2}[E_1|x,\sigma] \E_{C_1,\ell,C_2}[E_2|E_1,x,\sigma] \\ \geq &
	\frac{\epsilon}{2}\cdot\E_{C_1,\ell,C_2}[E_2|E_1,x,\sigma]. \tag*{(since $H$ occurs)}
	\end{align*}
We notice that if $E_1$ occurs, then $C_2\in\mathcal{C}_{x,\sigma}$, therefore
\[\E_{C_1,\ell,C_2}[T(C_1)_{|\ell}\neq T(C_2)_{|\ell}|C_2\in\mathcal{C}_{x,\sigma},x,\sigma] \leq \frac{2}{\epsilon}\cdot\E_{C_1,\ell,C_2}[E|x,\sigma] \leq \frac{2}{\epsilon}\frac{\epsilon}{2}\gamma\leq \gamma,\]
which means that $(x,\sigma)$ is $(\frac{\epsilon}{2}, \gamma)$ - excellent.
\end{proof}

For each $(x,\sigma)$ we define $f_{x,\sigma}$ by plurality over all cubes $C\in\cons_{x,\sigma}$.
\begin{definition}
For a pair $(x,\sigma)$ define a function $f_{x,\sigma} : \F^m \rightarrow \F$ as follows:
$$f_{x,\sigma}(y) = \mathop{\mathtt{argmax}}_{C\sim \mathcal{C}_y\cap\cons_{x,\sigma}}\left\{ T(C)(y)\right\} .$$
If $\mathcal{C}_y\cap\cons_{x,\sigma} = \emptyset$, define $f_{x,\sigma}(y)$ arbitrarily.
\end{definition}

\begin{claim}\label{claim:point equality}
For an $(\frac{\epsilon}{2}, \gamma)$ excellent pair $(x,\sigma)$,
$$\Pr_{C\sim \cons_{x,\sigma},y\sim C}[f_{x,\sigma}(y) = T(C)(y)] \geq 1-\gamma .$$
\end{claim}
\begin{proof}
Fix an $(\frac{\epsilon}{2}, \gamma)$ excellent pair $(x,\sigma)$, and denote $f=f_{x,\sigma}$.
If we pick a uniform $C_1\in\cons_{x,\sigma}$, then $y\in C_1$ such that $y\neq x$, and a uniform $C_2\in\cons_{x,\sigma}\cap \mathcal{C}_y$, then
\begin{align*}
\Pr_{C_1,y,C_2}[T(C_1)(y) \neq T(C_2)(y)] \leq
\Pr_{C_1,y,C_2}[T(C_1)_{|\ell(x,y)} \neq T(C_2)_{|\ell(x,y)}]\leq \gamma,
\end{align*}
since $(x,\sigma)$ is $(\frac{\epsilon}{2}, \gamma)$ excellent.
	
For each $y$, denote $\gamma_y =  \Pr_{C_1,C_2\sim\cons_{x,\sigma}\cap\mathcal{C}_y}[T(C_1)(y) \neq T(C_2)(y)]$. From the above we get that $\mathbb{E}_y[\gamma_y]\leq \gamma$, where $y$ is distributed according to it's weight in $\cons_{x,\sigma}$. For each $y$,
\begin{align*}
1-\gamma_y =& \sum_{\theta\in\F}\Pr_{C\sim\cons_{x,\sigma}\cap\mathcal{C}_y}[T(C)(y) = \theta]^2  \\\leq& \tag*{($f(y)$ is the most frequent value)}
\Pr_{C\sim\cons_{x,\sigma}\cap\mathcal{C}_y}[T(C)(y) = f(y)]\sum_{\theta\in\F}\Pr_{C\sim\cons_{x,\sigma}\cap\mathcal{C}_y}[T(C)(y) = \theta] \\\leq&
\Pr_{C\sim\cons_{x,\sigma}\cap\mathcal{C}_y}[T(C)(y) = f(y)].
\end{align*}
	
Since it is true for each $y$, it is also true when taking expectation over $y$, for any distribution:
\[\Pr_{C\sim\cons_{x,\sigma},y\sim C}[f(y) = T(C)(y)] = \Ex{y}{\E_{C\sim\cons_{x,\sigma}\cap\mathcal{C}_y}[\mathbb{I}(T(C)(y) = f(y))]} \geq \Ex{y}{1-\gamma_y}\geq1-\gamma .\]
In expectation, each $y$ is chosen with probability proportional to it's weight in $\cons_{x,\sigma}$, as before.
\end{proof}

\paragraph{Proof of \lref[Lemma]{lemma:local structure}:}
From \lref[Claim]{claim:good but not excellent} we know that the probability of $(x,\sigma)$ to be $(\frac{\epsilon}{2},\gamma)$-excellent is  at least $\frac{\epsilon}{3}$. Since $x$ is chosen uniformly, it means that for at least $\frac{\epsilon}{3}$ of the inputs $x\in\F^m$ there exists some $\sigma\in\F$ such that $(x,\sigma)$ is excellent. If there is more than one such $\sigma$ choose one arbitrarily.

Fixing an excellent $x$, let $\sigma$ be the value such that $(x,\sigma)$ is excellent. For this $\sigma$, $\Pr_{C\in \mathcal{C}_x}[C\in\cons_{x,\sigma}]\geq \frac{\epsilon}{2}$. From \lref[Claim]{claim:point equality}, $\Pr_{C\sim \cons_{x,\sigma},y\sim C}[f_{x,\sigma}(y) = T(C)(y)] \geq 1-\gamma $. By averaging, at least half of the cubes $C\in\cons_{x,\sigma}$ satisfy $\Pr_{y\sim C}[f_{x,\sigma}(y) = T(C)(y)] \geq 1-2\gamma $. For all these cubes $T(C)\approxparam{2\gamma}f_{x,\sigma}$, and they are at least $\frac{\epsilon}{4}$ fraction of the cubes in $\mathcal{C}_x$.
\qed

\subsection{Global Structure}\label{sec:global}
In this section, we prove the following lemma:
\begin{lemma}[Global Structure] \label{lemma:global structure}
Let $T$ be a cubes table that passes \lref[Test]{test:CxC} with probability at least $\epsilon =\Omega( \frac{d^4}{\sqrt{q}})$, then for every $\gamma = \Omega(\frac{1}{d^3})$, there exists an $(\frac{\epsilon}{2},\gamma)$-excellent $x$ such that $f = f_x:\F^m\rightarrow \F$ satisfies
\[ \Pr_{C}[T(C) \approxparam{32\gamma} f_{|_C}] \geq \frac{\epsilon}{16}.\]
\end{lemma}
Let $X^\star\subseteq\F^m$ the set of $(\frac{\epsilon}{2},\gamma)$ excellent points.

The main idea in the proof of the global structure, is showing that there exist many pairs of excellent points  $x,y\in X^\star$, such that for many cubes $C$, the $T(C)$ is similar both to $f_x$ and to $f_y$  (\lref[Claim]{claim:many interactions}). If this is the case, then the functions $f_x,f_y$ must be very similar (\lref[Claim]{claim:approx equality}). Finally, the lemma is proven by averaging and finding a single $x$ such that $f_x$ agrees simultaneously with many of the $f_y$'s and their supporting cubes.

\begin{definition}[Supporting cubes]
For any excellent $x\in X^\star$, we denote by $F_x$ the set of cubes ``supporting'' $f_x$,
\[ F_x = \left\{ C\in \mathcal{C}_x \left| T(C) \approxparam{2\gamma} f_{x|_C} \right. \right\} .\]
\end{definition}
\begin{claim}
\label{claim:many interactions}
Let $\D$ be the following process: choose $x,y\in X^\star$ independently and uniformly at random, let $C$ be a random cube containing both $x$ and $y$. Then
\[\Pr_{x,y,C\sim D}[C\in F_x\cap F_y]\geq \frac{\epsilon^2}{26} .\]
\end{claim}
\begin{proof}
Since each $x\in X^\star$ is excellent, we know from the local structure lemma, \lref[Lemma]{lemma:local structure}, that $\Pr_{C\sim \mathcal{C}_x}[C\in F_x]\geq \frac{\epsilon}{4}$. This is of course also true when taking a uniform $x\in X^\star$, thus, $\Pr_{x \sim X^\star, C\sim\mathcal{C}_x}[C\in F_x]\geq \frac{\epsilon}{4}$.

From \lref[Lemma]{claim:ev_FmC}
, the inclusion graph  $G=G(\F^m,\mathcal{C})$ has $\lambda(G) = \lambda \leq (1+o(1))\frac{1}{q^{3/2}}$. Denote the measure of $X^\star$ by $\mu$, from \lref[Lemma]{lemma:local structure}, $\mu\geq\frac{\epsilon}{3}$. Hence, by the application of
 \lref[Lemma]{lemma:edge sampling} on the graph $G$ with $A=\mathcal{C}, B=\F^m$ and $B' =X^\star $, we get
\begin{align} \label{eq:ex point cube}
\abs{\Pr_{x \sim X^\star, C\sim\mathcal{C}_x}[C\in F_x] - \Pr_{C\sim \mathcal{C},x\sim C\cap X^\star}[C\in F_x]} \leq \frac{\lambda}{\sqrt{\mu}}\leq \frac{2\lambda}{\sqrt{\epsilon}}.
\end{align}

For each $C\in\mathcal{C}$, let $p_C =\Pr_{x\sim C\cap X^\star}[C\in F_x]$, this measures for every cube $C$ how many points  $x\in C$ are such that $f_{x|_C}\approxparam{2\gamma} T(C)$. In this notation, \eqref{eq:ex point cube} implies $\E_C[p_C]\geq \frac{\epsilon}{4} -\frac{2\lambda}{\sqrt{\epsilon}} \geq  \frac{\epsilon}{5}$.
We can use this to bound the probability of the event $C\in F_x\cap F_y$ by first choosing $C$, then two independent points in $C\cap X^\star$,
\[\Pr_{\substack{C\sim\mathcal{C}\\x,y\sim C\cap X^\star}}[C\in F_x\cap F_y] =\E_C[p_C^2]\geq \left(\E_C[p_C] \right)^2\geq \frac{\epsilon^2}{25}  .\]

We observe that this distribution is very similar to the required distribution $D$. The only difference is that here we first pick $C\in \mathcal{C}$ and then two excellent points in $C$, whereas in $D$ we first pick two points in $X^\star$ and then a common neighbor $C$. The graph $G$ satisfies that every two distinct points $x,y\in\F^m$ have exactly the same number of common neighbors. Therefore, we can use \lref[Lemma]{lemma:two edges sampling} on the graph $G$ with $A=\mathcal{C}, B=\F^m$ and $B' =X^\star $ to get
\[ \abs{\Pr_{\substack{C\sim\mathcal{C}\\x,y\sim C\cap X^\star}}[C\in F_x\cap F_y] - \Pr_{x,y,C\sim D}[C\in F_x\cap F_y]} \leq \frac{2\lambda}{\mu} +\frac{1}{\mu^2d_A} + \frac{1}{\mu^2\abs{B}} \leq \frac{6\lambda}{\epsilon} + \frac{9}{q^m\epsilon^2} + \frac{9}{q^3\epsilon^2} .\]
Recall that $\lambda \leq (1+o(1))\frac{1}{q^{3/2}}$ and since $\epsilon= \Omega( \frac{d^4}{\sqrt{q}})$,
we conclude that $ \Pr_{x,y,C\sim D}[C\in F_x\cap F_y]\geq \frac{\epsilon^2}{25} - \frac{6\lambda}{\epsilon} -  \frac{9}{q^m\epsilon^2} - \frac{9}{q^3\epsilon^2} \geq \frac{\epsilon^2}{26}$.
\end{proof}

\begin{claim}\label{claim:approx equality}
Let $x\neq y\in X^\star$, and let $\ell$ be the line containing $x$ and $y$, if $\Pr_{C\sim \mathcal{C}_{\ell}}[C\in F_x\cap F_y] \geq \frac{\epsilon^2}{100}$ then $f_x\approxparam{5\gamma}f_y$.
\end{claim}
\begin{proof}
Consider the graph $G = G(\F^m\setminus\ell,\mathcal{C}_\ell)$. This is a bi-regular bipartite graph, and by \lref[Lemma]{claim:ev_FmC-ell}
it has $\lambda = \lambda(G)\leq (1+o(1))\frac{1}{\sqrt{q}}$.
Let $F = F_x\cap F_y$. By assumption, $F$ has measure at least $\frac{\epsilon^2}{100}$ inside $\mathcal{C}_\ell$.

We denote by $E'\subset E$ the edges of $G$ that indicate agreement with both $f_x$ and $f_y$,
\[ E' = \{ (z,C) \mid T(C)(z) = f_x(z) = f_y(z) \} .\]
Every cube $C\in F$ has $1-2\gamma$ of the points $z\in C$ satisfying $T(C)(z) = f_x(z)$ and $1-2\gamma$ of the points satisfying $T(C)(z)=f_y(z)$. By a union bound we get $\Pr_{C\in F,z\in N(C)}[(z,C)\in E']\geq 1-4\gamma$.
By \lref[Lemma]{lemma:edge sampling} on $G$ when $A= \F^m\setminus \ell, B=\mathcal{C}_\ell, B' = F$,
\[ \abs{\Pr_{C\sim F,z\sim N(C)}[(z,C)\in E'] - \Pr_{z,C\sim N(z)\cap F}[(z,C)\in E']  } \leq \frac{20\lambda}{\epsilon},\]
which means that $\Pr_{z\sim \F^m,C\sim N(z)\cap F}[(z,C)\in E']\geq 1-4\gamma - \frac{20\lambda}{\epsilon}\ge 1-5\gamma$. By the definition of $E'$, for each point $z\in\F^m$ that has an adjacent edge in $E'$, $f_x(z) = f_y(z)$.
This means that
\[ \Pr_z[f_x(z) = f_y(z)]\geq \Pr_z[\exists C \st (z,C)\in E'] \geq \Pr_{z,C\sim N(z)\cap F}[(z,C)\in E'] \geq 1-5\gamma.  \]
\end{proof}
The above claim showed that if two functions have a large set of  cubes on which they almost agree then these functions are similar. In order to prove the global structure, we also need to show that in this case, most of $C\in F_y$ will also be close to $f_x$.
\begin{claim}\label{claim:common agreement}
Let $x,y\in X^\star$ such that $f_x\approxparam{5\gamma}f_y$, then
\[ \Pr_{C\sim F_y}[T(C)\approxparam{32\gamma} f_{x|_C}]\geq \frac{1}{2} .\]
\end{claim}
Note that the function $f_x$ may not be a low degree polynomial, so $T(C)\approxparam{32\gamma} f_{x|_C}$ doesn't imply equality.
\begin{proof}
Let $G = G(\F^m\setminus \{y\}, \mathcal{C}_y)$, by \lref[Claim]{claim:ev_FmCx}
it has $\lambda = \lambda(G) \approx \frac{1}{q}$. First, we denote by $E'_y$ the following set of edges,
\[ E'_y = \{(z,C) \mid T(C)(z) = f_y(z) \} .\]
For each $C\in F_y$, we know that $\Pr_{z\in N(C)}[(z,C)\in E'_y]\geq 1-2\gamma$.
From \lref[Lemma]{lemma:edge sampling} on $G$ when $A=\F^m\setminus y, B=\mathcal{C}_y, B' = F_y$,
we know that
\[ \abs{\Pr_{C\sim F_y, z\sim N(C)}[(z,C)\in E'_y]- \Pr_{z,C\in N(z)\cap F_y}[(z,C)\in E'_y]} \leq \frac{4\lambda}{\epsilon},\]
since the measure of $F_y$ is at least $\frac{\epsilon}{4}$. This implies that $\Pr_{z,C\in N(z)\cap F_y}[(z,C)\in E'_y]\geq 1-3\gamma$.

We define a second set of edges, $E_x'$ to be the same only for $f_x$,
\[ E'_x = \{(z,C) \mid T(C)(z) = f_x(z) \} .\]
We notice that if $z$ is a point such that $f_x(z) = f_y(z)$, then $(z,C)\in E_y'\implies (z,C)\in E_x'$.
\begin{align*}
\Pr_{z,C\sim N(z)\cap F_y}[(z,C)\in E'_x]\geq& \Pr_{z}[f_x(z) = f_y(z)] \cdot \Pr_{z,C\sim N(z)\cap F_y}[(z,C)\in E'_y\mid f_x(z) = f_y(z)] \\\geq&
(1-5\gamma)\cdot \Pr_{z,C\sim N(z)\cap F_y}[(z,C)\in E'_y\mid f_x(z) = f_y(z)] \tag*{(since $f_x\approxparam{5\gamma}f_y$)} \\\geq&
(1-5\gamma)\cdot \left(\Pr_{z,C\sim N(z)\cap F_y}[(z,C)\in E'_y] - 5\gamma\right)\\\geq&
1-15\gamma.
\end{align*}
Therefore, we can use \lref[Lemma]{lemma:edge sampling}
again on the same graph $G$ and set $F_y$, now with the edge set $E_x'$, to conclude that
\[ \Pr_{C\sim F_y, z\sim N(C)}[(z,C)\in E'_x] \geq \Pr_{z,C\sim N(z)\cap F_y}[(z,C)\in E'_x] - \frac{4\lambda}{\epsilon} \geq 1-16\gamma,\]
By averaging, at least half of $C\in F_y$ satisfies $T(C)\approxparam{32\gamma} f_{x|_C}$.
\end{proof}
We are now ready to prove the global structure.
\paragraph{Proof of \lref[Lemma]{lemma:global structure}:}
Let $T$ be the cubes table that passes \lref[Test]{test:CxC} with probability at least $\epsilon=\Omega( \frac{d^4}{\sqrt{q}})$. From the local structure, \lref[Lemma]{lemma:local structure}, we know that there exists a set $X^\star$ of excellent points, such that each $x\in X^\star$ has a function $f_x$, and $\abs{F_x}\geq\frac{\epsilon}{4}\abs{\mathcal{C}_x}$.

From \lref[Claim]{claim:many interactions}, we know that $\Pr_{x,y,C\sim D}[C\in F_x\cap F_y]\geq \frac{\epsilon^2}{26}$, when $x,y$ are chosen uniformly from $X^\star$ and $C$ is a common neighbor. Therefore, there must be $x\in X^\star$ such that $\Pr_{y\sim X^\star,C\sim N(x)\cap N(y)}[C\in F_x\cap F_y]\geq \frac{\epsilon^2}{26}$.

Fix such $x\in X^\star$, and let $X'$ be the set of $y\in X^\star$ such that $\abs{F_x\cap F_y}\geq \frac{\epsilon^2}{100}\abs{\mathcal{C}_\ell}$. By averaging,  $\abs{X
'}\geq \frac{\epsilon^2}{100}\abs{X^\star} \geq \frac{\epsilon^3}{400}\abs
\F^m$.

By \lref[Claim]{claim:approx equality}, for all $y\in X'$, $f_y\approxparam{5\gamma}f_x$. For each $y\in X'$, let
\[ F'_y = \{C\in F_y \mid T(C) \approxparam{32\gamma} f_{x|_C} \} .\]
At this point we have a large collection of $y$'s and for each one a large collection of cubes $F'_y$ such that all of these support the same function $f_x$. It is immediate that $f_x$ is supported by some $\poly(\eps)$ fraction of all of the cubes. Since we are aiming for a better quantitative bound of $\Omega(\eps)$ fraction of $\mathcal{C}$, we will rely on the expansion once more.

In order to finish the proof, we need to show that $\abs{\cup_{y\in X'}F'_y} \geq\frac{\epsilon}{16}\abs{\mathcal{C}}$.

Let $G=G(\F^m,\mathcal{C})$, by \lref[Lemma]{claim:ev_FmC}
$\lambda(G)\leq q^{-\frac{3}{2}}$. We use $X'$ as the set of vertices, and define \[E' = \{ (y,C) \mid T(C) \approxparam{32\gamma} f_{x|_C} \}.\]

By \lref[Lemma]{lemma:edge sampling} on $G$ with $A = \mathcal{C}, B=\F^m, B' = X' $,
\[\abs{\Pr_{y\sim X',C\sim N(y)}[(y,C)\in E'] -\Pr_{C\sim \mathcal{C}, y\sim N(C)\cap X'}[(y,C)\in E'] } \leq \frac{20\lambda}{\sqrt{\epsilon^3}} \leq \frac{20 q^{-\frac{3}{2}}}{q^{-\frac{3}{4}}} \leq 20q^{-\frac{3}{4}}\leq \frac{\epsilon}{16}, \]
where we used the fact that $\epsilon\geq \frac{1}{\sqrt{q}}$.

\lref[Claim]{claim:common agreement} lets us bound the first term on the left, since for each $y\in X'$, $\Pr_{C\sim N(y)}[C\in F'_y]\geq \frac{1}{2}\Pr_{C\sim N(y)}[C\in F_y] \geq \frac{\epsilon}{8}$. Thus,
\[ \Pr_{C\sim \mathcal{C}, y\sim N(C)\cap X'}[(y,C)\in E']  \geq \frac{\epsilon}{8}-\frac{\epsilon}{16} =\frac{\epsilon}{16}  .\]
We notice that a cube with even a single adjacent edge in $E'$ satisfies $ T(C) \approxparam{32\gamma} f_{x|_C}$, so we are done.
\qed

\subsection{Low Degree}\label{sec:lowdeg}

The last step is to prove that the global function discovered in the previous section can be modified to make it a low degree function, while still maintaining large support for it among the cubes.

\begin{theorem}[Theorem \ref{thm:main} restated]\label{thm:main restated}
For every $d$ and large enough prime power $q$ and every $m\ge 3$ the following holds. Let $T$ be a cubes table that passes \lref[Test]{test:CxC} with probability at least $\epsilon = \Omega(\frac{d^4}{\sqrt{q}})$, then there exist a degree $d$ polynomial $g:\F^m\to\F$ such that $T(C) = g|_C$ on an $\Omega(\eps)$ fraction of the cubes.
\end{theorem}
From \lref[Lemma]{lemma:global structure}, we get a function $f$ such that $\Omega(\epsilon)$ of the cubes have $T(C)\approx f_{|_C}$. In this section, we will show that this function $f$ is close to  a degree $d$ polynomial $g$. Afterwards, we also need to show that $\Omega(\epsilon)$ of the cubes satisfies $T(C)=g_{|_C}$

To show the first part, we will use a robust characterization of low degree polynomials given by  Rubinfeld and Sudan.
\begin{theorem}[{\cite[Theorem 4.1]{RS96}}]
\label{thm:RuSu}
Let $f:\F^m\rightarrow\F$ be a function,  and let $N_{y,h} = \{y + i(h-y) \mid i\in \{0,\dots,d+1\} \}$, if $f$ satisfies
\[ \Pr_{y,h\in\F^m}[\exists \deg d \text{ polynomial }  p \st p_{|_{N_{y,h}}} = f_{|_{N_{y,h}}}] \geq 1-\delta, \]
for $\delta \leq \frac{1}{2(d+2)^2}$,  then there exists a degree $d$ polynomial $g$ such that $f\approxparam{2\delta}g$.
\end{theorem}
For completeness, we present proof of the above theorem in \lref[Appendix]{section:RuSu proof}.

\begin{claim}\label{claim:low deg}
Fix any $\gamma \leq \frac{1}{100(d+2)^3}$, let $f:\F^m\rightarrow\F$ and $x\in \F^m$ such that
$ \Pr_{C\in\mathcal{C}_x}[T(C)\approxparam{32\gamma}f_{|_C}]\geq \frac{\epsilon}{4} $, then exists a degree $d$ polynomial $g$ such that
$f\approxparam{84d\gamma} g$.
\end{claim}
\begin{proof}
Denote by $F\subseteq\mathcal{C}_x$ the following set
\[ F = \{C\in\mathcal{C}_x \mid T(C)\approxparam{32\gamma}f_{|_C} \} .\]
Our first goal is to show that for nearly all lines, $f$ agrees with a low degree function on almost all of the points of the line.

Fix $C\in F$, if we pick a uniform $\ell\subset C$ we expect that $T(C)_\ell\approxparam{O(\gamma)}f_{|_\ell}$. Using the spectral properties we show that almost all lines satisfy this property. Let $G_C=G(A\cup B,E)$ be the following bipartite inclusion graph where $A$ is all the points in $C$, and $B$ is all the affine lines in $C$. Let $A'\subset A$ be $A'= \{ y\in A \mid T(C)(y)\neq f(y) \}$, and $B'\subset B$ be $B'= \{\ell \in B \mid \abs{N(\ell)\cap A'} \geq 40\gamma\abs{N(\ell)} \}$.
From \lref[Lemma]{claim:ev_FmL} with $m=3$ (we apply the lemma where "$\F^m$" is the cube $C$), $\lambda_C = \lambda(G_C)\leq \frac{2}{\sqrt{q}}$.
We apply \lref[Lemma]{lemma:edge sampling} on $G_C$ and the set $B'$, where the set of edges is all the edges adjacent to $A'$:
\[ \abs{\Pr_{y\in A,\ell\in N(y)\cap B'}[y\in A'] - \Pr_{\ell\in B',y\in N(\ell)}[y\in A']} \leq \frac{\lambda_C}{\sqrt{\frac{\abs{B'}}{\abs{B}}}} .\]
We notice that $\Pr_{y\in A}[y\in A']\leq 32\gamma$. By the definition of $B'$, $ \Pr_{\ell\in B',y\in N(\ell)}[y\in A'] \geq 40\gamma$.  Therefore $\abs{B'}\leq \left(\frac{\lambda_C}{8\gamma}\right)^2\abs{B} < \gamma\abs{B}$.

We have shown that for every cube $C\in F$, almost all lines in it satisfy $T(C)_\ell\approxparam{40\gamma}f_{|_\ell}$. Now we need to show that the set $F$ is large enough to cover $(1-O(\gamma))$ of all the lines in $\mathcal{L}$.
The inclusion graph $G=G(\mathcal{L}\setminus\mathcal{L}_x,\mathcal{C}_x)$ has $\lambda = \lambda(G) \leq \frac{1}{\sqrt{q}}$, by \lref[Lemma]{claim:ev_LCx}.
We denote by $E'$ the set of edges $(\ell,C)$ such that $T(C)_{|_\ell} \approxparam{40\gamma}f_{|_\ell}$. As we've seen above, for every $C\in F$, $\Pr_{\ell\in N(C)}[(\ell,C)\in E']\geq 1-\gamma$.

By \lref[Lemma]{lemma:edge sampling} on $G$, with $A=\mathcal{L}\setminus\mathcal{L}_x, B=\mathcal{C}_x, B' = F$,
\begin{align*}
\abs{\Pr_{\ell,C\sim N(\ell)\cap F}[(\ell,C)\in E']-\Pr_{C\sim F, \ell\sim C}[(\ell,C)\in E']} \leq \frac{\lambda}{\sqrt{\epsilon}} \leq \gamma,
\end{align*}
which means that
\begin{align*}
\Pr_{\ell}[\exists C \st (\ell,C)\in E'] \geq \Pr_{\ell,C\sim N(\ell)\cap F}[(\ell,C)\in E'] \geq 1-2\gamma.
\end{align*}
This means that for $1-2\gamma$ of the lines in $\mathcal{L}$, $f$ agrees with a degree $d$ function on $1-40\gamma$ fraction of the points of each line.

We are very close to being able to apply the low degree test of Rubinfeld and Sudan \cite{RS96}, that works in the high soundness regime. For this, we need to move to neighborhoods.
For $y,h\in\F^m$, we define the neighborhood of $y,h$, \[N_{y,h} = \{y+i(h-y)\mid 0\leq i\leq d+1\}.\] Notice that $N_{y,h}\subset\ell(y,h)$.
We show that on almost all of the neighborhoods $N_{y,h}$, the function $f_{|_{N_{y,h}}}$ equals a degree $d$ polynomial, by showing that for almost all $N_{y,h}$, there exists some cube $C$ such that $f_{|_{N_{y,h}}}=T(C)_{|_{N_{y,h}}}$ ($T(C)$ is a degree $d$ polynomial).

Picking a random neighborhood $N_{y,h}$ is equivalent to picking a random line $\ell\in\mathcal{L}$ and then uniform $y,h\in\ell$. We have already showed that almost all lines $\ell\in\mathcal{L}$, there exists a cube $C$ such that $T(C)_\ell\approxparam{\Omega(\gamma)}f_{|_\ell}$.

Now we can bound the same probability over neighborhoods
\begin{align}
\Pr_{y,h\sim\F^m}[\exists C \st f(N_{y,h}) = T(C)(N_{y,h})] \geq& \Pr_{\ell}[\exists C \st (\ell,C)\in E'] \cdot \nonumber\\ &\Pr_{\ell,y,h\sim\ell}[f(N_{y,h}) = T(C)(N_{y,h}) \mid \exists C \st (\ell,C)\in E'] \nonumber\\ \geq&
(1-2\gamma)\Pr_{\ell,y,h\sim\ell}[f(N_{y,h}) = T(C)(N_{y,h}) \mid \exists C \st (\ell,C)\in E'] \nonumber\\
\geq& (1-2\gamma)(1-(d+2)\cdot40\gamma), \label{eq:union bound}\\
\geq& 1-42d\gamma, \nonumber
\end{align}
where (\ref{eq:union bound}) is due to union bound on the neighborhoods inside $\ell$.
Therefore, the function $f$ equals a degree $d$ polynomial on $(1-42d\gamma)$ of the neighborhoods. Since $\gamma \leq 100(d+2)^{-3}$, by \lref[Theorem]{thm:RuSu}, we get that there exists a degree $d$ polynomial $g$, such that $f\approxparam{84d\gamma}g$.
\end{proof}

\paragraph{Proof of \lref[Theorem]{thm:main restated}:}
Fix the cubes table $T$, and let $f:\F^m\rightarrow\F$ be the function promised from \lref[Lemma]{lemma:global structure}. This function satisfies the conditions of \lref[Claim]{claim:low deg}, so there exists a degree $d$ polynomial $g$ such that $f\approxparam{84d\gamma}g$.

Since $g$ is a degree $d$ polynomial, for every cube $C$ either $T(C)=g_{|_C}$, or else they are very different.
Let $G$ be the inclusion graph $G=G(\F^m,\mathcal{C})$, and let
\[ F = \{ C\in\mathcal{C} \mid T(C)\approxparam{32\gamma}f_{|_C} \} \]
From \lref[Lemma]{lemma:global structure}, the measure of $F$ is at least $\frac{\epsilon}{16}$, let $A'$ be the set of points on which $f\neq g$.
By \lref[Lemma]{claim:ev_FmC}, $\lambda(G)\leq q^{-\frac{3}{2}}$. We use \lref[Lemma]{lemma:edge sampling} on $G$ with $A=\F^m, B=\mathcal{C}, B' = F$,
\[ \abs{ \Pr_{C\in F,y\in N(C)}[y\in A'] - \Pr_{y,C\in N(y)\cap F}[y\in A'] } \leq \frac{q^{-\frac{3}{2}}}{\epsilon} \leq \gamma \]
We know that $\Pr_{y,C\in N(y)\cap F}[y\in A'] \leq \Pr_{y}[y\in A'] \leq 84d\gamma$, which implies that $\Pr_{C\in F,y\in N(C)}[y\in A']\leq 85d\gamma$.

By averaging, for at least half of the cubes $C\in F$, $\Pr_{y\in C}[y\in A']  \leq 200d\gamma\leq  \frac{1}{2}$.
For all these cubes $T(C) = g_{|_C}$, because  $\Pr_{y\in C}[T(C)(y) = g(y)]\geq \Pr_{y\in C}[T(C)(y) = f(y),y\notin A']\geq 1-32\gamma - \frac{1}{2}> d/q$, and since $g_{|_C},T(C)$ are both degree $d$ polynomials, they must be equal.
\qed

\begin{remark}
\label{remark:improve d}Instead of \lref[Theorem]{thm:RuSu}, we can use another similar characterization from \cite{RS96}, where the neighborhood is defined as $N_{y,h} = \{y + i(h-y) \mid i\in \{0,\dots,10d\} \}$. The advantage of using this new neighborhood is that we can conclude $f\approxparam{(1+o(1))\delta} g$ as long as  $\delta = O(1/d)$. This will help in reducing the exponent of $d$ by $1$ in our main theorem.  We chose to use \lref[Theorem]{thm:RuSu} for a self contained proof.
\end{remark}

\newcommand{\sub}[1]{\mathcal{A}^{#1}}
\section{Comparing between different tests and their agreement parameter}\label{sec:tests}
There are many variants for the low degree test, in this section we look into equivalences between similar low degree agreement tests. We first prove the equivalence in a more general setting and as a corollary we get some interesting results.

Throughout this section, we will work over $\F^m$ where $\F$ is a field of size $q$ and let $\LL\leq m/2$ be fixed. Also, let $T$ denotes a table which maps every $\LL$ dimensional affine subspace in $\F^m$ to a degree $d$ polynomial. Let $\sub{\LL}$ denote the set of all $\LL$ dimensional affine subspaces in $\F^m$. For $\RL<\LL$ and for $\rs\in \sub{\RL}$ let $\sub{\LL}_{\rs}\subseteq \sub{\LL}$ denote all subspaces in $\sub{\LL}$ which contain a particular subspace $\rs$,
\[ \sub{\LL}_{\rs} = \left\{ \ls\subset\F^m \mid \dim(\ls)=\LL, \rs\subseteq \ls \right\} .\]
For parameters $\LL>\KL\geq \RL$ consider the following test:

\begin{test}
	\caption{Subspace agreement test : $\alpha_{\LL \KL \LL(\RL)}$}
	\label{test:subspaces_gen}
	\begin{enumerate}
		\item Select $\ks \in \sub{\KL}$ u.a.r.
		\item Pick $\ls_1, \ls_2 \in \sub{\LL}_\ks$ u.a.r.
		\item Pick a $\RL$ dimensional subspace $\rs\subseteq \ks$ u.a.r.		
		\item Accept iff $T(\ls_1)_{|\rs} = T(\ls_2)_{|\rs}$.
	\end{enumerate}
Let $\alpha_{\LL \KL \LL(\RL)}(T)$ be the {\em agreement} of the table $T= (f_{\ls})_{\ls\in \sub{\LL}}$, i.e. the probability of acceptance of the test.
\end{test}

 When $\RL=\KL$ we simply denote the agreement as $\alpha_{\LL \KL \LL}(T)$. With these notations, the success probability of \lref[Test]{test:CxC} is denoted by $\alpha_{3,0,3}(T)$, and of \lref[Test]{test:PlP} by $\alpha_{2,1,2}(T)$.

In this section, we prove the following main lemma.
\begin{lemma}
\label{lemma:test equivalence}
Let $0\leq  \RL <  \KL< \LL \leq \frac{m}{2}$, we have

$$\alpha_{\LL \RL \LL}(T)\left(1 - \left(\frac{d}{q}\right)^{\RL+1}\right) \leq \alpha_{\LL \KL \LL}(T) \leq \alpha_{\LL \RL \LL}(T) + (1+o(1))q^{-(\LL-2\KL+\RL+1)},$$
\end{lemma}
From \lref[Lemma]{lemma:test equivalence}, we can deduce the following corollary,
\begin{corollary}
Let $\aclc(T) =\alpha_{3,1,3}(T) $ be the success probability of \lref[Test]{test:subspaces_gen} with $\LL=3,\KL=\RL=1$, i.e checking consistency of two cubes that intersect on a line. Then for every cubes table $T$,
 \[ \acxc(T)\left(1 - \frac{d}{q}\right)\leq \aclc(T) \leq \acxc(T) + \frac{1}{q^2}(1+o(1)). \]
\end{corollary}
The corollary implies that \lref[Theorem]{thm:main} holds if we modify the test as selecting two cubes u.a.r from a pair of cubes intersecting in a line and checking consistency on the whole line.

Using \lref[Lemma]{lemma:test equivalence}, we can also compare the Raz-Safra Plane vs. Plane agreement tests where planes intersect at a point and on a line. Recall that $\aplp(T)$ is the acceptance probability of \lref[Test]{test:PlP}. Invoking \lref[Lemma]{lemma:test equivalence} with $\LL=2$, $\KL=1$ and $\RL=0$, we get the following corollary.
\begin{corollary}[{\lref[Lemma]{lem:pxp plp}} restated]
Let $T$ be a planes table, and let $\apxp(T)$ be the success probability of \lref[Test]{test:subspaces_gen} with $\LL=2, \KL=\RL=0$, i.e two planes that intersects on a point. Let $ \aplp(T)$ be the success probability of  \lref[Test]{test:PlP} from the introduction (two planes that intersects on a line), then
 \[ \apxp(T)\left(1 - \frac{d}{q}\right)\leq \aplp(T) \leq \apxp(T) + \frac{1}{q}(1+o(1)). \]
\end{corollary}

\subsection{Proof of \lref[Lemma]{lemma:test equivalence}}
We prove a few claims that together with the observation $\alpha_{\LL \KL \LL(\RL)}(T)\geq\alpha_{\LL \KL \LL}(T)$, prove the lemma.

The following claim shows that two distinct low degree polynomials agree on a random subspace of fixed dimension with very small probability.
\begin{claim}
\label{claim:extSZ}
Let $P_1, P_2 : \F^t \rightarrow \F$ be two distinct degree $d$ polynomials. For $\RL\leq t$
$$\Pr_{\rs\in \sub{\RL}}\left[(P_1)_{|\rs} \equiv (P_2)_{|\rs}\right] \leq \left(\frac{d}{q}\right)^{\RL+1}.$$
\end{claim}
\begin{proof}
Consider the following way of choosing an $\RL$ dimensional {\em affine} subspace from $\sub{\RL}$ uniformly at random: Pick $x_0, x_1, x_2, \ldots, x_\RL$ from $\F_q^t$ independently and u.a.r. Then pick a $\RL$ dimensional affine subspace $\rs$ containing $\{x_0 + \mathtt{span}(x_1, x_2, \ldots, x_\RL)\}$ u.a.r ($\rs$ is determined by $x_0, x_1, x_2, \ldots, x_\RL$, unless $\dim{\mathtt{span}(x_1, x_2, \ldots, x_\RL)}< \RL$). It is easy to see that $\rs$ is distributed uniformly in $\sub{\RL}$. Now, $P_1$ and $P_2$ agreeing on the whole subspace $\rs$ implies that they agree on the points $\{x_0, x_0+x_1, x_0+x_2, \ldots, x_0 + x_\RL\}$ as all these points are contained in $\rs$. Therefore,
\begin{align*}
\Pr_{\rs\in \sub{\RL}}[(P_1)_{|\rs} \equiv (P_2)_{|\rs}]
&\leq \Pr_{x_0, x_1, x_2, \ldots, x_\RL\sim \F^t}\left[P_1(x_0) = P_2(x_0) \wedge_{i=1}^{\RL} P_1(x_0+x_i) = P_2(x_0+x_i)\right]  \\
& = \left(\Pr_{x\in \F_q^t}\left[P_1(x) = P_2(x)\right]\right)^{\RL+1}\leq \left(\frac{d}{q}\right)^{\RL+1},
\end{align*}
where the last inequality is because two different degree $d$ polynomial agree on at most $\frac{d}{q}$ fraction of the points (Schwartz-Zippel lemma).
\end{proof}
\begin{claim}\label{claim:bound inner product}
Let $M_{m\times n}$ be the adjacency matrix of a bi regular bipartite graph $G$, and let $f$ be a $n$-dimensional $\{0,1\}$ vector such that $\E[f] = \mu$. Then
\[\langle Mf,Mf \rangle \leq \mu^2+\lambda(G)^2\mu. \]
\end{claim}
\begin{proof}
Let $\one$ be the unit vector. We write $f$ as $f = f_1 + f_1^\perp$ where $f_1$ is in the direction of $\one$, the singular vector with the maximal singular value, and $f_1^\perp$ is its orthogonal component. We note that $f_1 = \mu \one$, and hence $\langle f_1,f_1 \rangle = \mu^2$. Also,
\[\mu = \langle f,f \rangle = \langle f_1+f_1^\perp,f_1+f_1^\perp \rangle = \langle f_1,f_1 \rangle + \langle f_1^\perp,f_1^\perp \rangle\geq \langle f_1^\perp,f_1^\perp \rangle  .\]
Using this we can bound:
\begin{align*}
\langle Mf,Mf \rangle =&
\langle Mf_1+Mf_1^\perp,Mf_1+Mf_1^\perp) \rangle \\=&
\langle f_1,f_1 \rangle +\langle Mf_1^\perp,Mf_1^\perp \rangle \\
\leq&\mu^2 +\lambda(G)^2\langle f_1^\perp,f_1^\perp \rangle \\
\leq& \mu^2  + \lambda(G)^2\mu.
\end{align*}
\end{proof}

\begin{claim}\label{claim:inter geq point}
 $\alpha_{\LL \KL\LL(\RL)}(T)\geq\alpha_{\LL \RL\LL}(T)$.
\end{claim}
\begin{proof}
 We start by fixing $\rs\in\sub{\RL},\sigma\in\mathbb{F}^{q^\RL}$.
For each $\KL$ dimensional subspace $\ks\in\sub{\KL}_\rs$, denote by $p_\ks$ the following probability $p_\ks=\Pr_{\ls\sim\sub{\LL}_\ks}[T(\ls)_{|\rs}\equiv \sigma]$. In this notation
\begin{equation}
\label{eq:convexity}
\Pr_{\substack{\ks\sim\sub{\KL}_\rs\\ \ls_1,\ls_2\sim \sub{\LL}_\ks}}[T(\ls_1)_{|\rs} \equiv  T(\ls_2)_{|\rs}\equiv \sigma] = \E_\ks[p_\ks^2]\geq \left(\E_\ks[p_\ks] \right)^2  = \Pr_{\ls_1,\ls_2\sim \sub{\LL}_\rs}[T(\ls_1)_{|\rs} \equiv  T(\ls_2)_{|\rs}\equiv \sigma].
\end{equation}

Now, we average over $\rs,\sigma$ to get $\alpha_{\LL \RL\LL}(T)$ and $\alpha_{\LL \KL\LL(\RL)}(T)$:
\begin{equation}
\label{eq:LRL}
\alpha_{\LL \RL\LL}(T) = \Pr_{\substack{\rs\sim\sub{\RL}\\ \ls_1,\ls_2\sim \sub{\LL}_\rs}}[T(\ls_1)_{|\rs} \equiv  T(\ls_2)_{|\rs}] =
\E_{\rs\sim\sub{\RL}}\left[\sum_{\sigma\in\mathbb{F}^{q^\RL}}\Pr_{\ls_1,\ls_2\sim \sub{\LL}_\rs}[T(\ls_1)_{|\rs} \equiv  T(\ls_2)_{|\rs}\equiv \sigma] \right] .
\end{equation}
Picking a uniform $\rs\in\sub{\RL}$ then $\ks\in\sub{\KL}_\rs$ is the same as picking $\ks\in\sub{\KL}$ and then a random $\RL$ dimensional subspace $\rs$ in $\ks$, so by definition
\begin{equation}
\label{eq:LKLR}
\alpha_{\LL \KL\LL(\RL)}(T) =
\Pr_{\substack{\rs\sim\sub{\RL},\ks\sim\sub{\KL}_\rs\\  \ls_1,\ls_2\sim \sub{\LL}_\ks}}[T(\ls_1)_{|\rs} \equiv  T(\ls_2)_{|\rs}] =
\E_{\rs\sim\sub{\RL}}\left[\sum_{\sigma\in\mathbb{F}^{q^\RL}}\Pr_{\substack{\ks\sim\sub{\KL}_\rs\\ \ls_1,\ls_2\sim \sub{\LL}_\ks}}[T(\ls_1)_{|\rs} \equiv  T(\ls_2)_{|\rs}\equiv \sigma] \right].
\end{equation}
Using (\ref{eq:convexity}), (\ref{eq:LRL}) and (\ref{eq:LKLR}), we get $\alpha_{\LL \KL\LL(\RL)}(T)\geq\alpha_{\LL \RL\LL}(T)$.
\end{proof}

\begin{claim}\label{claim:line geq inter}
$\alpha_{\LL\KL\LL}(T)\geq\alpha_{\LL\KL\LL(\RL)}(T)\left(1 - \left(\frac{d}{q}\right)^{\RL+1}\right)$.
\end{claim}
\begin{proof}
By the definition of the agreement,
\[\alpha_{\LL\KL\LL}(T) = 1- \E_{\ks\sim \sub{\KL}}\left[\Pr_{\ls_1,\ls_2\sim \sub{\LL}_\ks}[T(\ls_1)_{|_\ks}\neq T(\ls_2)_{|_\ks}] \right], \]

and
\[\alpha_{\LL\KL\LL(\RL)}(T) = 1- \E_{\ks\sim \sub{\KL}}\left[\Pr_{\substack{\rs\sim\ks,\\ \ls_1,\ls_2\sim\sub{\LL}_\ks}}[T(\ls_1)_{|\rs}\neq T(\ls_2)_{|\rs}]\right], \]
where we use $\rs\sim \ks$ to denote a random $\RL$ dimensional subspace in $\ks$.  For every subspace $\ks\in\sub{\KL}$, $\rs\subseteq\ks$ is uniform and is independent of $\ls_1,\ls_2$.
\begin{align*}
\Pr_{\substack{\rs\sim\ks,\\ \ls_1,\ls_2\sim\sub{\LL}_\ks}}[T(\ls_1)_{|\rs}\neq T(\ls_2)_{|\rs}] =& \Pr_{\substack{\rs\sim\ks,\\ \ls_1,\ls_2\sim\sub{\LL}_\ks}}[T(\ls_1)_{|\ks}\neq T(\ls_2)_{|\ks},T(\ls_1)_{|\rs}\neq T(\ls_2)_{|\rs}]\\
=& \Pr_{\ls_1,\ls_2\sim\sub{\LL}_\ks}[T(\ls_1)_{|\ks}\neq T(\ls_2)_{|\ks}]\cdot\\ &
\Pr_{\substack{\rs\sim\ks,\\ \ls_1,\ls_2\sim\sub{\LL}_\ks}}[T(\ls_1)_{|\rs}\neq T(\ls_2)_{|\rs}\mid T(\ls_1)_{|\ks}\neq T(\ls_2)_{|\ks}]\\
\geq& \Pr_{\ls_1,\ls_2\sim\sub{\LL}_\ks}[T(\ls_1)_{|\ks}\neq T(\ls_2)_{|\ks}]\cdot\left(1 - \left(\frac{d}{q} \right)^{\RL+1}\right).
\end{align*}
The lower bound on the probability in the last inequality is as follows: the event $T(\ls_1)_{|\ks}\neq T(\ls_2)_{|\ks}$ implies that the degree $d$ polynomials corresponding to $T(\ls_1)_{|\ks}$ and $ T(\ls_2)_{|\ks}$ are distinct.  Thus, using \lref[Claim]{claim:extSZ} $\Pr_{\rs\sim \ks}[T(\ls_1)_{|\rs}\equiv  T(\ls_2)_{|\rs}]\leq (d/q)^{\RL+1}$.
Therefore, for a $\KL$ dimensional subspace $\ks\in\sub{\KL}$,
\[\Pr_{\substack{\rs\sim\ks,\\ \ls_1,\ls_2\sim\sub{\LL}_\ks}}[T(\ls_1)_{|\rs}\neq T(\ls_2)_{|\rs}] \geq  \Pr_{\ls_1,\ls_2\sim\sub{\LL}_\ks}[T(\ls_1)_{|\ks}\neq T(\ls_2)_{|\ks}] \left(1 - \left(\frac{d}{q}\right)^{\RL+1}\right).\]
Finally, taking the expectation of the inequality over $\ks$ finishes the proof.
\end{proof}


We first state a lemma about an expansion of the kind of inclusion graphs which we will be dealing with in analyzing the \lref[Test]{test:subspaces_gen}, the proof of which appears in \lref[Appendix]{section:inclusion proofs}.
\begin{lemma}\label{lemma:subspaces}
	Let $\RL\leq \KL< \LL \leq \frac{m}{2}$ be integers, and let $G$ be the inclusion graph $G=G(\sub{\KL}_\rs, \sub{\LL}_\rs)$ for a $\RL$ dimensional subspace $\rs$, where $\rs\neq\emptyset$. Then,
	\[ \lambda(G)^2 \leq  (1+o(1))\cdot {q^{-(\LL-2\KL+\RL+1)}}.\]
\end{lemma}

\begin{claim} \label{claim:pointwise point geq line}
$\alpha_{\LL\KL\LL(\RL)}(T)\leq \alpha_{\LL\RL\LL(\RL)}(T) + \lambda(G)^2$
where $G$ is the inclusion graph $G=G(\sub{\KL}_\rs, \sub{\LL}_\rs)$ for an $\RL$ dimensional subspace $\rs$.
\end{claim}

\begin{proof}
Fix an $\RL$ dimensional affine subspace $\rs\in \sub{\RL}$. We prove the following inequality:
\begin{equation}
\label{eq:almostclaim}
\Pr_{\substack{\ks\sim \sub{\KL}_\rs,\\\ls_1,\ls_2\sim \sub{\LL}_\ks}}[T(\ls_1)_{|\rs} \equiv T(\ls_2)_{|\rs}  ] \leq \Pr_{\ls_1,\ls_2\sim\sub{\LL}_\rs}[T(\ls_1)_{|\rs} \equiv T(\ls_2)_{|\rs}   ]  + \lambda(G)^2,
\end{equation}
Note that this implies the claim if we take expectation over $\rs\in \sub{\RL}$. Towards proving (\ref{eq:almostclaim}), for each value $\sigma\in\F^{q^k}$, denote by $A_\sigma\subseteq \sub{\LL}_\rs$ the following set
\[ A_\sigma = \{\ls\in \sub{\LL}_\rs \mid T(\ls)_{|\rs} \equiv \sigma \} ,\]
and $\mu_\sigma = \frac{\abs{A_\sigma}}{\abs{\sub{\LL}_\rs}}$.
Let $f_\sigma$ be the indicator function for $A_\sigma$, for $\ls\in A_\sigma$, $f_\sigma(\ls) = 1$.
By definition
\begin{equation}
\label{eq:point_test}
\Pr_{\ls_1,\ls_2\sim\sub{\LL}_\rs}[T(\ls_1)_{|\rs} \equiv T(\ls_2)_{|\rs}   ] = \sum_{\sigma}\mu_\sigma^2.
\end{equation}

Let $G=G(\sub{\KL}_\rs, \sub{\LL}_\rs)$ be the inclusion graph, and denote by $M\in\mathbb{R}^{|\sub{\KL}_\rs|\times  |\sub{\LL}_\rs|}$ the normalized adjacency matrix, such that each entry is either $0$ or $\frac{1}{\deg(\ks)}$ where $\ks\in \sub{\KL}_\rs$.

For each $\KL$ dimensional subspace $\ks\in \sub{\KL}_\rs$, the value $(Mf_\sigma)_\ks$ is the fraction of $\ks$'s neighbors in $A_\sigma$, $(Mf_\sigma)_\ks =\Pr_{\ls\sim \sub{\LL}_\ks}[\ls\in A_\sigma]$. Therefore, the inner product gives us the expected value:
\[ \langle Mf_\sigma,Mf_\sigma \rangle = \E_{\ks\in \sub{\KL}_\rs}\left[\E_{\ls\in\sub{\LL}_\ks}[\ls\in A_\sigma]^2\right] = \E_{\ks\in\sub{\KL}_\rs}\left[\E_{\ls_1,\ls_2\in\sub{\LL}_\ks}[\ls_1,\ls_2\in A_\sigma] \right] .\]

Therefore
\begin{align*}
\Pr_{\substack{\ks\sim \sub{\KL}_\rs,\\\ls_1,\ls_2\sim\sub{\LL}_\ks}}[T(\ls_1)_{|\rs} \equiv T(\ls_2)_{|\rs}  ] &=\sum_{\sigma}\langle Mf_\sigma,Mf_\sigma \rangle \\
&\leq \sum_{\sigma} \mu_\sigma^2 + \lambda(G)^2 \mu_\sigma \tag*{(using \lref[Claim]{claim:bound inner product})}\\
&= \Pr_{\ls_1,\ls_2\sim\sub{\LL}_\rs}[T(\ls_1)_{|\rs} \equiv T(\ls_2)_{|\rs}  ] + \lambda(G)^2 \tag*{(from (\ref{eq:point_test}) )}.
\end{align*}
which proves (\ref{eq:almostclaim}).
\end{proof}

\lref[Claim]{claim:pointwise point geq line} together with \lref[Lemma]{lemma:subspaces} gives us $\alpha_{\LL \KL \LL}(T) \leq \alpha_{\LL \RL \LL}(T) + (1+o(1))q^{-2(\LL-2\KL+\RL+1)}$. \lref[Claim]{claim:inter geq point} and \lref[Claim]{claim:line geq inter} prove the other inequality, $\alpha_{\LL \RL \LL}(T)\left(1 - \left(\frac{d}{q}\right)^{\RL+1}\right) \leq \alpha_{\LL \KL \LL}(T)$.

\appendix

\section{Spectral properties of Certain Inclusion Graphs}
\label{section:inclusion proofs}

Let $G_{\LL,\KL}$ be the intersection graph where the vertex set is all {\em linear subspaces} of dimension $\LL$ in $\F_q^m$ and $U\sim U'$  iff $\dim(U\cap U') = \KL$.  We will use the $T_{\LL,\KL}$ to denote the {\em Markov operator} associated with a random walk on this graph. We will need following fact about eigenvalues of $T_{\KL, \KL-1}$.

\begin{definition}
$\KL$-th $q$-ary Gaussian binomial coefficient ${m \brack \KL}_q$ is given by
$$ {m \brack \KL}_q :=\prod_{i=0}^{\KL-1} \frac{q^m-q^i}{q^\KL-q^i}. $$
\end{definition}
As $q$ is fixed throughout the article, we will omit the subscript from now on.
\begin{fact}(~\cite[Theorem 9.3.3]{brouwer1989distance})
\label{fact:ev_grassmann}
Suppose $1\leq \KL\leq \frac{m}{2}$,
\begin{enumerate}
\item The number of $\KL$ dimensional linear subspaces in $\F_q^m$ is exactly ${m \brack \KL}$.

\item The degree of $G_{\KL,\KL-1}$ is $q{\KL\brack 1}{m-\KL\brack 1}$.

\item The eigen values of $T_{\KL,\KL-1}$ are
$$\lambda_j(T_{\KL,\KL-1}) = \frac{q^{j+1}{\KL-j \brack 1}{m-\KL-j \brack 1} - {j \brack 1}}{q{\KL\brack 1}{m-\KL\brack 1}},$$
with multiplicities ${m\brack j}-{m\brack j-1}$ for $j=0,1,\ldots, \KL$. Asymptotically, $\lambda_j(T_{\KL,\KL-1})  = \Theta(q^{-j})$.

\end{enumerate}
\end{fact}

\begin{claim}
\label{claim:ev_relation of grassmanns}
For any $1\leq \KL\leq \frac{m}{2}$ and , we have $\abs{\lambda_1(T_{\KL,\KL-2}) - \lambda_1(T_{\KL,\KL-1})^2}= (1+o(1)) \frac{1}{q^\KL}.$
\end{claim}
\begin{proof}
Consider a two-step random walk on the graph $G_{\KL,\KL-1}$. We will show that with very high probability, a two-step random walk on $G_{\KL,\KL-1}$ corresponds to a single step random walk on $G_{\KL,\KL-2}$. Let $U_1, U_2, U_3$ be the vertices from a two-step random walk on $G_{\KL,\KL-1}$. Note that conditioned on the event $\dim(U_1\cap U_3) = \KL-2$, the distribution of $(U_1, U_3)$  is exactly same as a single step random walk on $G_{\KL,\KL-2}$. We will upper bound the probability of the event $\dim(U_1\cap U_3) \neq \KL-2$.

Let $w_1=U_1\cap U_2$ and $w_2 = U_2\cap U_3$, we can describe the distribution of the two-step random walk as follows:
\begin{enumerate}
	\item Choose a uniform $k$ dimensional subspace $U_2$.
	\item Choose two random $k-1$ dimensional subspaces, $w_1,w_2\subset U_2$.
	\item Choose a point $x_1\in\F^m\setminus U_2$, and set $U_1=\text{span}(w_1,x_1)$.
	\item Choose a point $x_2\in\F^m\setminus U_2$, and set $U_3=\text{span}(w_2,x_2)$.
\end{enumerate}
By definition, $U_2$ has ${\KL\brack \KL-1}$ subspaces of size $k-1$, therefore $\Pr_{w_1,w_2}[w_1=w_2]=\frac{1}{{\KL\brack \KL-1}}$. In order to satisfy $\dim(U_1\cap U_3) \neq \KL-2$ given that $w_1\neq w_2$, the point $x_2$ should be in $U_1$. There are $q^k-q^{k-1}$ points in $U_1\setminus U_2$, and therefore this probability equals $\frac{\abs{U_1\setminus U_2}}{\abs{\F^m\setminus U_2}} = \frac{q^\KL-q^{\KL-1}}{q^m-q^\KL}$.
\begin{align*}
\Pr[\dim(U_1\cap U_3) \neq \KL-2] &= \Pr[w_1 = w_2] + \Pr[\dim(U_1\cap U_3) \neq \KL-2 \wedge w_1\neq w_2]\\
&= \frac{1}{{\KL\brack \KL-1}} + \left(1-\frac{1}{{\KL\brack \KL-1}}\right)\Pr[\dim(U_1\cap U_3) \neq \KL-2 \mid w_1\neq w_2]\\
&= \frac{1}{{\KL\brack \KL-1}} + \left(1-\frac{1}{{\KL\brack \KL-1}}\right)\cdot \frac{q^\KL-q^{\KL-1}}{q^m-q^\KL} =: \beta.
\end{align*}
Thus, we have
$$ T_{\KL,\KL-1}^2 = \beta \mathcal{N} + (1-\beta)T_{\KL,\KL-2},$$
where $\mathcal{N}$ is a Markov operator corresponding to the two-step random walk on $G_{k,k-1}$, conditioning on $\dim(U_1\cap U_3)\neq \KL-2$. The claim follows as $\beta = (1+o(1))1/q^\KL$.
\end{proof}

Following fact follows from the definition of $\lambda(G)$.
\begin{fact}
\label{fact:one vs two steps}
For a bi-regular bipartite graph $G(A,B)$, if $T$ is a Markov operator associated with a random walk of length two starting from $A$ (or $B$) then $\lambda(G)^2 = \lambda(T)$.
\end{fact}

We now prove  {\lref[Lemma]{lemma:ev_bounds}}.
\begin{lemma} [Restatement of {\lref[Lemma]{lemma:ev_bounds}}]We have for every $m\geq 6$,
\begin{enumerate}
\item For $G_1(\mathcal{L}\setminus \mathcal{L}_x, \mathcal{C}_x)$ , $\lambda(G_1) \approx \frac{1}{\sqrt{q}} $.
\item For $G_2(\mathcal{L}_x,\mathcal{C}_x))$ , $\lambda(G_2) \approx  \frac{1}{q}$.
\item For $G_3(\F^m\setminus \ell,\mathcal{C}_\ell)$ , $\lambda(G_3)\approx  \frac{1}{\sqrt{q}} $.
\item For $G_4(\F^m,\mathcal{C})$ , $\lambda(G_4) \approx  \frac{1}{q^{3/2}}  $.
\item For $G_5(\F^m\setminus\{x\},\mathcal{C}_x)$ , $\lambda(G_5) \approx  \frac{1}{q}  $.
\setcounter{listcounter}{\value{enumi}} 
\end{enumerate}
And for every $m\geq 3$
\begin{enumerate}
	\setcounter{enumi}{\value{listcounter}}
	\item For $G_6(\F^m,\mathcal{L})$, $\lambda(G_6)\approx \frac{1}{\sqrt{q}}$.
\end{enumerate}
where $\approx$ denotes equality up to a multiplicative factor of $1\pm o(1)$.
\end{lemma}
\begin{proof}
Suppose $T$ is an $n\times n$ Markov operator which is a convex combination of a bunch of other Markov operators: $T = \sum_{i=1}^k \alpha_i T_i$ where $\alpha_i\geq 0$ and $\sum_{i=1}^k \alpha_i = 1$, and that both $T$ and $T_i$'s are regular.
As the row sum of each Markov operator is $1$, the largest eigenvalue is $1$, since both $T$ and $T_i$'s are regular, the eigenvector of the largest eigenvalue is the all $1$ vector. The second largest eigenvalue of $T$ can be upper bounded by
\begin{align*}
\lambda(T)& :=  \max_{\substack{v\in \R^n , \|v\|=1, \\ v\perp \mathbf{1}} }\|Tv\| \\
& = \max_{\substack{v\in \R^n , \|v\|=1, \\ v\perp \mathbf{1}} }\left|\left|\sum_{i=1}^k \alpha_i T_i\right|\right|\\
&\leq \sum_{i=1}^k \max_{\substack{v\in \R^n , \|v\|=1, \\ v\perp \mathbf{1}} }\|\alpha_i T_i\| = \sum_{i=1}^k \alpha_i \lambda(T_i).
\end{align*}
In proving the lemma, we repeatedly use the above simple fact to upper bound the eigenvalue.
\begin{enumerate}
\item
Without loss of generality, we can assume $x={\bf 0}$. Let $d_L$ and $d_R$ denote the left and right degree of $G_1$ respectively. Fix a line $\ell$, $d_L$ is the number of cubes containing $\ell$ and not passing through ${\bf 0}$. Every point $x\notin\text{span}(\ell,\bf{0})$ defines a cube $C=\text{span}(x,\bf{0},\ell)$. Thus, the number of linear cubes containing $\ell$ equals $d_L = \frac{q^m-q^2}{q^3-q^2}$, where the denominator is the overcounting factor, the number of points that give the same cube.

Fix a linear cube $C$. The right degree is the number of lines in $C$ not passing through the origin which is $\frac{ {q^3\choose 2} }{ {q\choose 2} } -\frac{q^3-1}{q-1}$, where the first term counts all possible lines in $C$ (each two different points define a line, we divide by the double counting) and the second term counts all the lines in $C$ that pass through the origin.

Let $T_1$ be the Markov operator associated with a two-step random walk in $G_1$ starting from $\mathcal{C}_x$. Using \lref[Fact]{fact:one vs two steps}, in order to bound $\lambda(G_1)$ it is enough to bound the second largest eigenvalue of $T_1$. Since $G_1$ is bi-regular, the first eigenvector of $T_1$ is the all ones vector. For every cube $C$, the number of two-step walks starting from $C$ is $d_L\cdot d_R$.

If $\dim \{C_1\cap C_2\}= 1$, then the two cubes intersection is only on a line. Since both cubes are linear, it means that this line goes through the origin, therefore it doesn't correspond to a vertex on the left side, and there is no walk $C_1\rightarrow\ell\rightarrow C_2$, so $(T_1)_{C_1,C_2}=0$. Of course, the same holds if $\dim \{C_1\cap C_2\}=0$.

If $\dim \{C_1\cap C_2\}= 2$, there there is a plane going through the origin in both $C_1,C_2$. The number of walks $C_1\rightarrow\ell\rightarrow C_2$ equals the number of lines in this plane that don't contain the origin,  $\bf{0}$. Each pair of distinct points on the plane correspond to a line, and we divide by the double counting. Therefore the number of lines in a plane equals $\frac{ {q^2\choose 2} }{ {q\choose 2} }$. We subtract from it the number of lines in a plane that contains $\bf{0}$, resulting in $\frac{ {q^2\choose 2} }{ {q\choose 2} } - \frac{q^2-1}{q-1} =: \beta$.

If $C_1=C_2$, then exists a path $C_1\rightarrow\ell\rightarrow C_2$ for every line $\ell$ adjacent to $C_1$, and there are $d_R$ such lines.

Since $T_1$ is a Markov operator, we need to normalize the number of paths between $C_1,C_2$ by dividing in the total number of outgoing paths from $C_1$, which equals $d_R\cdot d_L$. Therefore,
 \begin{equation}
  {(T_1)}_{C_i, C_j}=
    \begin{cases}
      {\frac{d_R}{d_R\cdot d_L}}, & \text{if}\ C_i=C_j \\
      \frac{\beta}{d_R\cdot d_L}, & \text{if}\ \dim \{C_1\cap C_2\}= 2\\
      0, &\text{otherwise}
    \end{cases}
  \end{equation}
Thus, we can write $T_1$ as:
$$ T_1 = \frac{1}{d_L} I + \frac{\beta}{d_R d_L} \cdot G_{3,2} = \frac{1}{d_L} I +  \frac {\beta d'}{d_R d_L}\cdot T_{3,2},$$
where $d'$ is the degree of a vertex in $G_{3,2}$. One can verify that $T_1$ is indeed a convex combination of two Markov operators $I$ and $T_{3,2}$. Since $G_{3,2}$ is a regular graph, the second eigenvector of $T_{3,2}$ is also orthogonal to $\bf{1}$. Hence,
\begin{align}
\lambda(G_1)^2 &= \lambda(T_1) = \max_{\substack{v\in \R^{|\mathcal{C}_x|}, v\perp \one\\\norm{v} = 1}} \norm{T_1 v}
= \max_{\substack{v\in \R^{|\mathcal{C}_x|}, v\perp \one\\\norm{v} = 1}} \norm{\left( \frac{1}{d_L} I +  \frac{\beta d'}{d_R d_L}\cdot T_{3,2}\right)v}\nonumber\\
&=  \frac{1}{d_L}+  \frac{\beta d'}{d_R d_L}\cdot \lambda_1(T_{3,2}). \label{eq:ev L,Cx expression}
\end{align}
We now just need to plug in the values of $ \beta , d'$ and $\lambda_1(T_{3,2})$. Using \lref[Fact]{fact:ev_grassmann}, $\lambda_1(T_{3,2})$ is given by the following expression,
$$\lambda_1(T_{3,2}) = \frac{q^2{2\brack 1} {m-4 \brack 1} - {1\brack 1}}{q{3\brack 1}{m-3\brack 1}} = (1+o(1))\frac{1}{q}.$$

As we have seen before, $d_R = \frac{ {q^3\choose 2} }{ {q\choose 2} } -\frac{q^3-1}{q-1} = (1+o(1))q^4$, $d_L = \frac{q^m-q^2}{q^3-q^2}=(1+o(1))q^{m-3}$ and $\beta = \frac{ {q^2\choose 2} }{ {q\choose 2} } - \frac{q^2-1}{q-1} =(1+o(1))q^2$. From \lref[Fact]{fact:ev_grassmann}, $d'=(1+o(1))q^{m-1}$.
Thus,
$$\frac{1}{d_L} = (1+o(1))\frac{1}{q^{m-3}}, \quad \quad \frac{\beta d'}{d_R d_L}\lambda_1(T_{3,2})= (1+o(1))\frac{1}{q} $$
Plugging these values in \eqref{eq:ev L,Cx expression} gives
 $\lambda(G_1) = (1+o(1))\frac{1}{\sqrt{q}}$ as required.

\item This bound is implied from a more general  \lref[Lemma]{lemma:subspaces} we prove below with $\LL = 3,\KL = 1$ and $\RL = 0$. 

\item In this case, it will be easier to bound the eigenvalue of the Markov operator associated with a random walk of length two starting from $\F^m\setminus \ell$. Let $T_3$ be the Markov operator. Now, the path of length two starting from $x$ looks like $x\rightarrow C\rightarrow y$. Thus, the cube $C$ contains all points from the affine plane spanned by $x$ and $\ell$. Let $p(x,\ell)$ be the affine plane spanned by $x$ and $\ell$.  We have $\Pr[y\in p(x,\ell)] = \frac{q^2-q}{q^3-q}\approx \frac{1}{q}$. If $y\notin p(x,\ell)$ then the distribution of $y$ is uniform in $\F^m\setminus p(x,\ell)$. Thus, we have
$$T_3 = (1-o(1))\left(1-\frac{1}{q}\right)J + (1+o(1))\frac{1}{q}\mathcal{N},$$
where $J$ is a Markov operator associated with a complete graph on $\F^m\setminus \ell$, with self loops and $\mathcal{N}$ is an appropriate Markov operator. 
Thus, we have bound $\lambda(T_3) = (1+o(1))\frac{1}{q}$. Since $\lambda(G_3)^2 = \lambda(T_3)$, the bound follows.

\item Proof of this is along the same lines as (3). The Markov operator here (starting a walk from the left side) can be written as
$$T_4 =(1\pm o(1)) \frac{1}{q^3}I + \left((1\pm o(1))(1-\frac{1}{q^3})\right)J, $$
where $I$ is an identity matrix. Thus $\lambda(T_4) = (1\pm o(1))\frac{1}{q^3} =\lambda(G_4)^2$.

\item The proof of this item is also similar to (3), we look on the path of length $2$ starting from the left side, i.e $y\rightarrow C\rightarrow z$, and let $T_5$ be the Markov operator. Let $\ell(x,y)$ be the line spanned by $x,y$ (where $x$ is the fixed point, $G_5(\F^m\setminus\{x\},\mathcal{C}_x)$), then $\Pr[z\in \ell(x,y)] =  \frac{\abs{\ell(x,y)\setminus \{x\}}}{\abs{C\setminus\{x\}}}= \frac{q-1}{q^3-1}\approx\frac{1}{q^2}$, let $\mathcal{N}$ be the appropriate Markov operator of the event that $x,y,z$ are colinear, then
\[ T_5 = (1-o(1))\left(1-\frac{1}{q^2}\right)J+(1+o(1))\frac{1}{q^2}\mathcal{N} .\]
Here $J$ is the Markov operator of the complete graph on $\F^m\setminus\{x\}$.
Thus $\lambda(G_5)^2\approx \frac{1}{q^2}$.

\item Consider a two-step random walk in $G_6$, $x\rightarrow\ell \rightarrow y$. If we sample a random line through $x$ then conditioned on $y\neq x$, $y$ is uniformly distributed in $\F^m$. Thus, we can write the Markov operator $T$ associated with this process as:
$$T = \frac{1}{q} I + \left(1-\frac{1}{q}\right)T',$$
where $T'$ is a Markov operator associated with a random walk on a complete graph on $A$, without self loops and $I$ is an identity matrix. As $T' = \frac{1}{|A|-1} J - \frac{1}{|A|-1}I$, $\lambda(T') = \frac{1}{q^3-1}$.
Thus, $\abs{\lambda(T) -\frac{1}{q}} \leq \frac{1}{q^3-1}$. The claim follows as $\lambda(G_6)^2 = \lambda(T)$.
\end{enumerate}
\end{proof}

Next, we prove \lref[Lemma]{lemma:subspaces}. Recall that  $\sub{\LL}$ denotes set of all $\LL$ dimensional affine subspaces in $\F^m$. Also, for $\RL<\LL$ and for $\rs\in \sub{\RL}$,   $\sub{\LL}_{\rs}\subseteq \sub{\LL}$ denotes all those subspaces in $\sub{\LL}$ which contains a particular subspace $\rs$.
\begin{lemma}[Restatement of {\lref[Lemma]{lemma:subspaces}}]
	Let $\RL\leq \KL< \LL \leq \frac{m}{2}$ be integers, and let $G$ be the inclusion graph $G=G(\sub{\KL}_\rs, \sub{\LL}_\rs)$ for an $\RL$ dimensional subspace $\rs$, where $\rs\neq\emptyset$. Then,
	\[ \lambda(G)^2 \leq  (1+o(1))\cdot {q^{-(\LL-2\KL+\RL+1)}}.\]
\end{lemma}
\begin{proof}
Fix an $\RL$ dimensional subspace $ \rs\subseteq \F^m, \rs\neq\emptyset$ and recall that
\[ \mathcal{A}_\rs^\KL = \left\{ \ks\subset\F^m | \dim(\ks) = \KL, \rs\subset \ks \right\} .\]
Let $G=G(\mathcal{A}_\rs^\KL,\mathcal{A}_\rs^\LL)$ be the biregular bipartite inclusion graph and let $d_\KL$ (resp. $d_\LL$) denote the degree of vertex in $\mathcal{A}_\rs^\KL$ (resp. $\mathcal{A}_\rs^\LL$).

For every $n,t,j\in\mathbb{N}$, let $h(n,t,j)$ be the number of $t$ dimensional subspaces in $\F^n$ that contain a specific dimention $j$ subspace,
\begin{equation}
\label{eq:boundsonh}
 h(n,t,j) = \frac{(q^n - q^j)\cdots(q^n - q^{t-1})}{(q^t - q^j)\cdots(q^t - q^{t-1})}\approx q^{(n-t)(t-j)},
 \end{equation}
where $\approx$ denotes equality up to a multiplicative factor $(1\pm o(1))$, as before. For any fixed $j$ dimensional subspace $X$, the numerator equals the number of $t-j$ linearly independent points $y_1, y_2,\ldots, y_{t-j}$ in $\F^n$ such that $\dim({\mathtt{span}}(X, y_1, y_2, \ldots, y_{t-j})) = t$, whereas for every $t$ dimensional subspace $Z$, the denominator equals the double counting of $Z$, i.e the number of $t-j$ linearly independent points $y_1 , y_2, \ldots, y_{t-j}$ such that ${\mathtt{span}}(X, y_1, y_2, \ldots, y_{t-j}) = Z$. We can now bound the number of vertices and the left and right degree in $G$.
\[ \begin{array}{cc}
\abs{\mathcal{A}_\rs^\KL} = h(m,\KL,\RL),  &\quad\quad \abs{\mathcal{A}_\rs^\LL} = h(m,\LL,\RL),  \\
d_\KL = h(m,\LL,\KL), &\quad\quad  d_\LL = h(\LL,\KL,\RL).
\end{array}
\]
%

Let $T$ be the two-step Markov operator on the bipartite graph  $G$, starting from $\mathcal{A}^\KL_\rs$, we want to calculate the entries of $T$. Let $\ks_1, \ks_2 \in \mathcal{A}_\rs^\KL$, by definition $(T)_{\ks_1,\ks_2}$ is the probability that a two-step random walk will end at $\ks_2$, conditioned on it starting from $K_1$.

Let $\RL' = \dim(\ks_1\cap \ks_2) \geq \RL$, in this notation $\dim(\ks_1\cup  \ks_2) = 2\KL-\RL'$. Any $2$ step random walk from $\ks_1$ to $\ks_2$ looks like $\ks_1 \rightarrow \ls' \rightarrow \ks_2$ where $\ls'$ is an $\LL$ dimentional subspace containing both $\ks_1$ and $\ks_2$. The number of such $\ls'$ is exactly $h(m, \LL, 2\KL-\RL')$. Thus, $(T)_{\ks_1,\ks_2}$ equals
\begin{equation}
\label{eq:pr-prime}
(T)_{\ks_1,\ks_2}=\Pr[\text{R.W ends at } \ks_2 | \text{ R.W starts at } \ks_1  ] = \frac{h(m, \LL, 2\KL-\RL')}{d_\KL\cdot d_\LL} = \frac{h(m, \LL, 2\KL-\RL')}{h(m, \LL, \KL) \cdot h(\LL, \KL, \RL)}.
\end{equation}
This probability is the same for every $\ks_1, \ks_2 \in \mathcal{A}_\rs^\KL$ such that $\dim(\ks_1\cap \ks_2)=\RL' $, so we can denote this value by $p_{\RL'}=(T)_{\ks_1,\ks_2}$.  Notice that $ p_{\RL'} \geq p_{\RL}$ for every $\RL' \geq \RL$.

Let $G_{\RL'}$ be the graph with vertex set $\mathcal{A}_\rs^\KL$, where $\ks_1,\ks_2$ are connected by an edge if $\dim(\ks_1\cap\ks_2) = \RL'$. We also denote the $0/1$ adjacency matrix of graph $G_{\RL'}$ by $G_{\RL'}$. With these notations, the $2$ step Markov operator $T$ equals
\[ T = \sum_{\RL'=\RL}^{\KL}p_{\RL'}G_{\RL'}.  \]
Notice that this is not a convex combination, $\sum_{\RL'}p_{\RL'}\neq 1$, but rather $p_{\RL'}$ are the entries of $T$, and $G_{\RL'}$ are $0/1$ matrices.

Let $J$ be the all $1$ matrix, we know that $J=\sum_{\RL'=\RL}^{\KL}G_{\RL'}$. The first matrix in the sum $G_\RL$ is the only non sparse matrix, since for every subspace $\ks_1\in \mathcal{A}^\KL_r$, almost all other subspaces intersects with $\ks_1$ only in $\rs$. Therefore we can write $G_\RL = J- \sum_{\RL'=\RL+1}^{\KL}G_{\RL'}$, and get
\[ T = p_{\RL}J + \sum_{\RL'=\RL+1}^{\KL}(p_{\RL'}-p_\RL)G_{\RL'}. \]

Since $T$ is a Markov operator of a regular graph, the all $\bf 1$ vector is the vector with the maximal eigenvalue, which equals $1$. Since $G_{\RL'}$ are also regular graphs, $\bf 1$ is the vector with the maximal eigenvalue, which equals $\deg(G_{\RL'})$, which is the number of $\ks'\in\mathcal{A}_\rs^\KL$ such that $\dim(\ks\cap \ks')=\RL' $ (as the adjacency matrices are not normalized).
\begin{align*}
	\deg(G_{\RL'}) = & h(\KL,\RL',\RL)\cdot \frac{(q^m - q^\KL)\cdots(q^m - q^{2\KL-\RL'-1})}{(q^\KL - q^{\RL'})\cdots(q^\KL - q^{\KL-1})}\\
	\approx & q^{(\KL-\RL')(\RL'-\RL)}\cdot q^{(m-\KL)(\KL-\RL')} = q^{(\KL-\RL')(m-\KL+\RL'-\RL)}
\end{align*}
For every $\ks\in\mathcal{A}^\KL_\rs$, the factor $h(\KL,\RL',\RL)$ is the number of $\RL'$ dimensional subspace in $\ks$ that contain $\rs$, the second factor is the number of $\KL$ dimensional subspaces that intersect with $\ks$ only in a specific  $\RL'$ dimensional subspace.

Let $v$ be the normalized eigenvector of the second eigenvalue of $T$, this means that $v\perp\bf 1$ and $\norm{v} = 1$. Since $J$ is the all $1$ matrix, $Jv = 0$. We also know that for every $\RL' > \RL$, $\norm{G_{\RL'}v}\leq \deg(G_{\RL'})$, as it is true for every vector $v$.
\begin{align}
	\norm{Tv}  =& \norm{\sum_{\RL'=\RL+1}^{\KL}(p_{\RL'}-p_\RL)G_{\RL'}v}\nonumber\\
	\leq& \sum_{\RL'=\RL+1}^{\KL}(p_{\RL'}-p_\RL)\norm{G_{\RL'}v} \tag*{(triangle inequality)}\\
	\leq& \sum_{\RL'=\RL+1}^{\KL}p_{\RL'}\deg{(G_{\RL'})} \nonumber
\end{align}
For every $\RL'$, by using the expression for $p_{\RL'}$ from (\ref{eq:pr-prime}) and bounds on $h$ from (\ref{eq:boundsonh}) we get that
\[ p_{\RL'}\deg{(G_{\RL'})} \approx p_{\RL'}q^{(\KL-\RL')(m-\LL+\RL'-\RL)}\approx q^{-(\RL'-\RL)(\LL-2\KL+\RL')}. \]
Since $\RL' > \RL$, $(\RL'-\RL)(\LL-2\KL+\RL')$ is minimized when $\RL' = \RL+1$ and hence
\[\lambda(T) = \norm{Tv}\leq (1+o(1)) \sum_{\RL'=\RL+1}^{\KL}\frac{1}{q^{\LL-2\KL+\RL'}}\leq (1+o(1))\cdot \frac{1}{q^{\LL-2\KL+\RL+1}} .\]
The lemma statement now follows from the \lref[Fact]{fact:one vs two steps}.
\end{proof}

\section{Spectral Expansion Properties Proofs}
\label{section:missing proofs}

\begin{lemma} [Restatement of  {\lref[Lemma]{lemma:edge sampling}}]
Let $D_1, D_2$ as defined in \lref[Definition]{def:D1D2}. Let $G = (A\cup B,E)$ be a bi-regular bipartite graph, then for every subset $B'\subset B$ of measure $\mu>0$ and every $E'\subset E$
\[ \abs{ \Pr_{(a,b)\sim D_1}[(a,b)\in E'] - \Pr_{(a,b)\sim D_2}[(a,b)\in E']} \leq \frac{\lambda(G)}{\sqrt{\mu}} .\]
Where is $D_2$ returned $\fail$, we treat is as it is not in $E'$.
\end{lemma}

\begin{proof}
In the proof we represent both probabilities as an inner product, and then use $\lambda(G)$ to bound the difference. Let $M\in\mathbb{R}^{A\times B}$ the adjacency matrix of the graph $G$, normalized such that $M\one = \one$ (where the first $\one$ is of dimension $\abs{B}$ and the second of dimension $\abs{A}$).
We define the matrix $M'$ representing the subset of edges $E'$, $M'_{a,b} = M_{a,b} \cdot (\one_{E'})_{a,b}$.

Starting with the probability of $(a,b)\sim D_1$, the vector $M'\one_{B'}$ satisfies that for every $a\in A$, $(M'\one_{B'})_a = \Pr_{b\in N(a)}[(a,b)\in E',b\in B']$.
\begin{align*}
\langle \one ,M'\one_{B'} \rangle  =&
\E_{a\sim A}\left[ E_{b\sim N(a)}[\mathbb{I}((a,b)\in E',b\in B')] \right] \\=&
\Pr_{a\sim A,b\sim N(a)}[(a,b)\in E',b\in B'] \tag*{(using bi-regularity of $G$)}\\=&
\Pr_{b\sim B,a\sim N(b)}[(a,b)\in E',b\in B'] \\=&
\Pr_{b\sim B}[b\in B']\cdot \Pr_{b\sim B,a\sim N(b)}[(a,b)\in E'\mid b\in B']  \\=&
\mu \cdot \Pr_{(a,b)\sim D_1}[(a,b)\in E'].
\end{align*}
We now want to represent the second probability as an inner product. We define the vector $P\in [0,1]^A$ as follows, for each $a\in A$:
\begin{enumerate}
	\item If $N(a)\cap B' = \emptyset$, then $P_a = 0$.
	\item Else, $P_a= \Pr_{b\in N(a)}[(a,b)\in E' \mid b\in B']$.
\end{enumerate}
In this notation $\Pr_{(a,b)\sim D_2}[(a,b)\in E'] = \langle \one ,P \rangle$.

We now want to find a connection between the inner products.
If $P_a\neq 0$, then it defined as the conditional probability, and
\[ \Pr_{b\sim N(a)}[b\in B', (a,b)\in E'] =  \Pr_{b\sim N(a)}[b\in B']\Pr_{b\sim N(a)}[ (a,b)\in E'\mid b\in B'] =  \Pr_{b\sim N(a)}[b\in B']P_a .\]
If $P_a = 0$ then also  $\Pr_{b\sim N(a)}[b\in B', (a,b)\in E'] = 0$, and the above equality still holds.
We notice that $(M'\one_{B'})_a = \Pr_{b\in N(a)}[(a,b)\in E',b\in B']$ and $(M\one_{B'})_a = \Pr_{b\in N(a)}[b\in B']$, which means that for every $a\in A$, $(M'\one_{B'})_a = (M\one_{B'})_a P_a$ and
\[ \langle M\one_{B'} ,P \rangle  = \langle \one ,M'\one_{B'} \rangle .\]

Therefore we can express the difference between the two probabilities as
\begin{align}
\abs{\Pr_{(a,b)\sim D_1}[(a,b)\in E'] - \Pr_{(a,b)\sim D_2}[(a,b)\in E']} =& \abs{\frac{1}{\mu}
\langle \one ,M'\one_{B'} \rangle - \langle \one ,P \rangle} \label{eq:edge sampling final}\\ =&
\abs{\frac{1}{\mu} \langle M\one_{B'} ,P \rangle - \langle \one ,P \rangle} \nonumber\\ =&
\frac{1}{\mu} \abs{\langle M\one_{B'} -\mu\one,P \rangle} \nonumber\\ \leq &
\frac{1}{\mu}\norm{M\one_{B'} -\mu\one}\norm{P} \tag*{(By Cauchy Swartz)}
\end{align}
Since $P$ is a vector in $[0,1]$ and the inner product we use is expectation, $\norm{P}\leq 1$.
In order to finish the proof we need to bound the size of the vector
\[ M\one_{B'} -\mu\one = M\one_{B'} -\mu M\one = M(\one_{B'} - \mu\one) .\]
We notice that $\one_{B'}$ is a $\{0,1\}$ vector of measure $\mu$, so $\langle \one_{B'}, \one \rangle = \langle \one_{B'}, \one_{B'} \rangle = \mu$, and $(\one_{B'} - \mu\one) \perp \one_{B}$. By the definition of $\lambda(G)$, this means that
\[ \norm{M(\one_{B'} - \mu\one)}\leq \lambda(G)\norm{\one_{B'} - \mu\one} \leq \lambda\sqrt{\mu}. \]
We substitute the norm of the vector in equation (\ref{eq:edge sampling final}) and we are done.
\end{proof}

\begin{lemma}[Restatement of {\lref[Lemma]{lemma:two edges sampling}}]
	Let $D_3, D_4$ as defined in \lref[Definition]{def:D3D4}. Let $G=(A\cup B,E)$ be a bi-regular bipartite graph, such that every two distinct  $b_1,b_2\in B$ have exactly the same number of common neighbors (i.e for all distinct $b_1,b_2\in B$, $|N(b_1)\cap N(b_2)|$ is the same), and this number is non-zero. Then for every subset $B'\subset B$ of measure $\mu>0$ and every $E'\subset E$
	\[ \abs{\Pr_{a,b_1,b_2\sim D_3}[(a,b_1)(a,b_2)\in E'] - \Pr_{a,b_1,b_2\sim D_4}[(a,b_1)(a,b_2)\in E']} \leq \frac{2\lambda(G)}{\mu} +\frac{1}{\mu^2d_A} + \frac{1}{\mu^2\abs{B}} \]
	Where is $D_4$ returned $\fail$, we treat is as it is not in $E'$ and $d_A$ is the degree on $A$ side.
\end{lemma}
\begin{proof}
This proof is similar in spirit to the proof of \lref[Lemma]{lemma:edge sampling}, with more complication since the event contains two edges instead of a single one.

Let $M\in\mathbb{R}^{A\times B}$ the adjacency matrix of the graph $G$, normalized such that $M\one = \one$.
We denote by $M'$ the matrix that represents the edges in $E'$, i.e for each $a\in A,b\in B$, $M'_{a,b} = M_{a,b} \cdot (\one_{E'})_{a,b}$.

Starting from $D_3$, we first write the conditional probability
\begin{align}
\Pr_{\substack{b_1,b_2\\a\sim N(b_1)\cap N(b_2)}}[b_1,b_2\in B',(a,b_1),(a,b_2)\in E'] =& \label{eq:prelim equality}
\Pr_{b_1,b_2}[b_1,b_2\in B']\Pr_{a,b_1,b_2\sim D_3}[(a,b_1),(a,b_2)\in E']  \\=&
\mu^2\Pr_{a,b_1,b_2\sim D_3}[(a,b_1),(a,b_2)\in E'] \nonumber
.
\end{align}

We want to express the left side as an inner product, we notice that for each $a\in A$:
\[ (M' \one_{B'})_a = \E_{b\sim N(a)}[\mathbb{I}(b\in B', (a,b)\in E')]  .\]
Therefore the inner product satisfies
\begin{align}
\langle  M'\one_{B'},M'\one_{B'}\rangle = &\E_{a\sim A}\left[\E_{b_1,b_2\sim N(a)}[\mathbb{I}(b_1,b_2\in B', (a,b_1)(a,b_2)\in E')]\right] \label{eq:d1 inner product}\\=&
\Pr_{a\sim A,b_1,b_2\sim N(a)}[b_1,b_2\in B', (a,b_1)(a,b_2)\in E'] \nonumber
\end{align}
Since each two $b_1,b_2\in B$ has the same number of neighbors,
\[ \Pr_{\substack{a\sim A\\b_1\neq b_2\sim N(a)}}[b_1,b_2\in B', (a,b_1)(a,b_2)\in E'] = \Pr_{\substack{b_1\neq b_2\sim B\\a\sim N(b_1)\cap N(b_2)}}[b_1,b_2\in B', (a,b_1)(a,b_2)\in E'] .\]
We want to switch the expression in (\ref{eq:d1 inner product}) by the one is (\ref{eq:prelim equality}), we know that they are equal when $b_1\neq b_2$. But the probability of $b_1=b_2$ is different between the two cases, it is $\frac{1}{d_A}$ if we pick neighbors of $a$ and $\frac{1}{\abs{B}}$ if we pick two random vertices in $B$.
If we add the probability of $b_1=b_2$ as an error, we get that
\begin{align}
\abs{\mu^2\Pr_{a,b_1,b_2\sim D_3}[ (a,b_1)(a,b_2)\in E'] - \langle  M'\one_{B'},M'\one_{B'}\rangle} \leq \frac{1}{d_A} + \frac{1}{\abs{B}}
\label{eq:d1 bound}
\end{align}

Now we want to express the probability of $a,b_1,b_2\sim D_4$ as an inner product. In order to do that, we define the vector $P$, for every $a\in A$
\begin{enumerate}
	\item If $N(a)\cap B' = \emptyset$, then $P_a= 0$.
	\item Else, $P_a= \Pr_{b_1,b_2\sim N(a)}[(a,b_1)(a,b_2)\in E' \mid b_1,b_2\in B']$.
\end{enumerate}
The vector $P$ is defined such that
\[ \Pr_{a,b_1,b_2\sim D_4}[(a,b_1)(a,b_2)\in E'] = \E_a[P_a] =  \langle \one, P \rangle. \]

We want to find a connection between this expression and the expression representing the probability $ \Pr_{a,b_1,b_2\sim D_3}[(a,b_1)(a,b_2)\in E']$.

We use (\ref{eq:d1 bound}) and the triangle inequality to bound the difference between the two target probabilities
\begin{align}
\abs{\Pr_{a,b_1,b_2\sim D_3}[(a,b_1)(a,b_2)\in E'] - \Pr_{a,b_1,b_2\sim D_4}[(a,b_1)(a,b_2)\in E']}\leq& \abs{\frac{1}{\mu^2 }\langle  M'\one_{B'},M'\one_{B'}\rangle - \langle \one, P \rangle} \nonumber\\
&+ \frac{1}{\mu^2d_A} + \frac{1}{\mu^2\abs{B}} \label{eq:diff}
\end{align}

We now need to bound the expression in (\ref{eq:diff}), in order to do that, we will first show that
\begin{align}
\langle  M'\one_{B'},M'\one_{B'}\rangle=  \Pr_{a\sim A,b_1,b_2\sim N(a)}[(a_1,b)(a_2,b)\in E',  b_1,b_2\in B'] = \E_a[P_a (M\one_{B'})_a^2].\label{eq:connection}
\end{align}
We notice that for $a$ such that $P_a >0$, it equals the conditional probability and
\[ \Pr_{b_1,b_2\sim N(a)}[(a_1,b)(a_2,b)\in E' , b_1,b_2\in B'] = \Pr_{b_1,b_2\sim N(a)}[b_1,b_2\in B'] P_a.\]
If $a$ is such that $P_a = 0$, then $\Pr_{b_1,b_2\sim N(a)}[(a_1,b)(a_2,b)\in E' , b_1,b_2\in B'] = 0$ and the above equality still holds. We further notice that
\[(M\one_{B'})_a = \E_{b\sim N(a)}[\mathbb{I}(b\in B')].\]
If we substitute $\Pr_{b_1,b_2\sim N(a)}[b_1,b_2\in B']$ in $(M\one_{B'})_a^2$, we get (\ref{eq:connection}).

In order to finish the proof, we upper bound
\[ \abs{\frac{1}{\mu^2 }\langle  M'\one_{B'},M'\one_{B'}\rangle - \langle \one, P \rangle}  = \abs{\E_a\left[\frac{1}{\mu^2}P_a (M\one_{B'})_a^2 - P_a\right]} = \frac{1}{\mu^2}\abs{\E_a[P_a( (M\one_{B'})_a^2 - \mu^2)]}. \]

We now upper bound the expectation as follows,
\begin{align}
\E_a[P_a( (M\one_{B'})_a^2 - \mu^2)] =& \E_a[P_a((M\one_{B'})_a - \mu )((M\one_{B'})_a + \mu )] \nonumber\\ \leq& \max_{a}\{|P_a|\}\E_a[\abs{((M\one_{B'})_a - \mu )((M\one_{B'})_a + \mu )}]\nonumber\\ \leq&
\label{eq:cauchy swartz}
\norm{M\one_{B'} - \mu\one }\norm{M\one_{B'} + \mu \one} \\\leq&
\label{eq:sizes}
\lambda\sqrt{\mu} \sqrt{4\mu},
\end{align}
where (\ref{eq:cauchy swartz}) is due to Cauchy-Schwarz inequality and using $\abs{P_a}\leq 1$.
In (\ref{eq:sizes}), we bound

$\norm{M\one_{B'}- \mu \one }$ like in the previous proof,
\[ \norm{M\one_{B'}- \mu \one } = \norm{M\one_{B'}- \mu M\one } = \norm{M(\one_{B'} - \mu\one)}\leq \lambda \norm{\one_{B'}}\leq\lambda\sqrt{\mu} .\]
Finally, we bound  $\norm{M\one_{B'} + \mu\one}$:
\begin{align*}
\norm{M\one_{B'} + \mu\one }^2 =& \langle M\one_{B'} + \mu\one,M\one_{B'} + \mu\one\rangle \\=&
\langle M\one_{B'} ,M\one_{B'} \rangle +
2\langle M\one_{B'},\mu\one\rangle +
\langle \mu\one, \mu\one\rangle \\\leq&
\norm{\one_{B'}}^2 + 2\mu + \mu^2\norm{\one}^2 \\ \leq&
\mu + 2\mu + \mu^2 \leq 4\mu.
\end{align*}
\end{proof}

\newcommand{\maj}{\mathtt{maj}}
\section{Rubinfeld-Sudan Characterization}
\label{section:RuSu proof}
In this section, we present a proof of \lref[Theorem]{thm:RuSu}.  The proof uses the following fact from \cite{dW70}:

\begin{fact}
\label{fact:low deg char}
Let $f:\F^m\rightarrow\F$ be a function,  and let $N_{y,h} = \{y + ih \mid i\in \{0,\dots,d+1\} \}$. $f$ is degree $d$ iff it satisfies the following identity for all $y$ and $h$:
$$ \sum_{i=0}^{d+1} \alpha_i f(y+ih) = 0,$$
where $\alpha_i = {d+1\choose i}(-1)^{i+1}$.
\end{fact}
Throughout this section we let $\alpha_i = {d+1\choose i}(-1)^{i+1}$ as in the above fact.
\begin{theorem}[Restatement of {\lref[Theorem]{thm:RuSu}}]
Let $f:\F^m\rightarrow\F$ be a function,  and let $N_{y,h} = \{y + ih \mid i\in \{0,\dots,d+1\} \}$, if $f$ satisfies
\begin{equation}
\label{eq:locally low deg}
\Pr_{y,h\in\F^m}[\exists \deg d \text{ polynomial }  p \st p_{|_{N_{y,h}}} = f_{|_{N_{y,h}}}] \geq 1-\delta,
\end{equation}
for $\delta \leq \frac{1}{2(d+2)^2}$,  then there exists a degree $d$ polynomial $g$ such that $f\approxparam{2\delta}g$.
\end{theorem}

\begin{proof}
 Define a function $g:\F^m \rightarrow\F$ to be $g(y) = \maj_{h\in \F^m}\{\sum_{i=1}^{d+1} \alpha_i f(y+ih)\}$ breaking the ties arbitrarily. Next we argue that $g$ is very close to $f$ and $g$ itself is a degree $d$ function.

To see that $g$ is $(1-2\delta)$ close to $f$, consider the set of all $y$ for which $\Pr_h[f(y) = \sum_{i=1}^{d+1} \alpha_i f(y+ih)] >1/2$.  For all these $y$, $f(y)=g(y)$ as $g$ was the majority vote. It is easy to see that fraction of $y$ for which the probability is at most $1/2$ is at most $2\delta$ as otherwise it will contradict the hypothesis ~(\ref{eq:locally low deg}). The rest of the proof will be proving the following two claims.

\begin{claim}
\label{claim:agreement at y}
For all $y\in \F^m$,  $\Pr_h[g(y) = \sum_{i=1}^{d+1} \alpha_i f(y+ih)] \geq 1-2(d+1)\delta.$
\end{claim}

\begin{claim}
\label{claim:g is low deg}
For all $y$ and $h$ in $\F^m$, we have $\sum_{i=0}^{d+1}\alpha_i g(y+ih) = 0$.
\end{claim}

\lref[Claim]{claim:g is low deg}  and \lref[Fact]{fact:low deg char} imply that $g$ is in fact a degree $d$ function and hence the theorem follows. We now proceed with proving these two claims.

%

{\bf Proof of \lref[Claim]{claim:agreement at y}}: We will show that for all $y\in \F^m$,
\begin{equation}
\label{eq:doublecollision}
\Pr_{h_1, h_2}\left[ \sum_{i=1}^{d+1} \alpha_i f(y+ih_1) = \sum_{j=1}^{d+1} \alpha_j f(y+jh_2)\right] \geq 1-2(d+1)\delta.
\end{equation}
Note that this is enough to prove the claim. To see this, let $p_a = \Pr_h[\sum_{i=1}^{d+1} \alpha_i f(y+ih) =a]$ for $a\in \F$. Then ~(\ref{eq:doublecollision}) becomes $\sum_{a\in \F} p_a^2 \geq 1-2(d+1)\delta$. Since $g(y)$ was the majority vote, we have $\Pr_h[g(y) = \sum_{i=1}^{d+1} \alpha_i f(y+ih)] = \max_{a\in \F} p_a \geq \sum_{a\in \F}p_a^2 \geq 1-2(d+1)\delta$.

To prove~(\ref{eq:doublecollision}), consider the following $(d+2)\times(d+2)$ matrix $Z$ with $(i,j)^{th}$ entry $Z_{i,j} = \alpha_i\alpha_j f(y+ih_1+jh_2)$, for $i,j\in\{0,\dots,d+1\}$.
\[ Z = \begin{bmatrix}
    f(y) &  \dots  & \alpha_0\alpha_j f(y+ jh_2) & \ldots \\
    \vdots &  \ddots &  \vdots  & \ddots  \\
    \alpha_i \alpha_0f(y +ih_1)  & \ldots & \alpha_i\alpha_j f(y+i h_1 + jh_2) & \ldots \\
    \vdots  & \ddots  & \vdots & \ddots
\end{bmatrix}
\]
If $h_1\in \F^m$ u.a.r then for any $i\in \{1,2,\ldots, d+1\}$,  $ih_1$ is distributed uniformly in $\F^m$. Same is true for $h_2$ and $jh_2$. Consider the following events:
\begin{itemize}
\item For every $i\in \{1,2,\ldots, d+1\}$, $R_i$ be the event that the sum of the $i$'th row is $zero$, i.e $ \sum_{j=0}^{d+1} Z_{i,j} = 0$.
\item For every $j\in \{1,2,\ldots, d+1\}$, $C_j$ be the event that sum of the $j$'th column is $zero$, i.e $\sum_{i=0}^{d+1} Z_{i,j} = 0$.
\end{itemize}
Note that $R_i,C_j$ are not defined for the first row and column ($i=0$ and $j=0$).
Using the hypothesis ~(\ref{eq:locally low deg}) of the theorem and \lref[Fact]{fact:low deg char}, we have
\begin{align*}
\Pr_{h_1, h_2} [R_i] \geq 1-\delta, \quad \quad \quad \forall i\in \{1,2,\ldots, d+1\}\\
\Pr_{h_1, h_2} [C_j] \geq 1-\delta, \quad \quad \quad \forall j\in \{1,2,\ldots, d+1\}
\end{align*}

The event in ~(\ref{eq:doublecollision}) is same as $\sum_{i=1}^{d+1} Z_{i,0} = \sum_{j=1}^{d+1} Z_{0,j}$ (note that the sums don't include the first element, $Z_{0,0}$). If all the above events $R_i, C_j$ happen then $\sum_{i=1}^{d+1} Z_{i,0} = \sum_{j=1}^{d+1} Z_{0,j}=-\sum_{i,j = 1}^{d+1} Z_{i,j}$.  By using union bound we get $\Pr[\mathop{\wedge}_{i=1}^{d+1} R_i \mathop{\wedge}_{j=1}^{d+1} C_j ] \geq 1- 2(d+1)\delta$ which implies ~(\ref{eq:doublecollision}).

%
%

{\bf Proof of \lref[Claim]{claim:g is low deg}}: In this case, consider the following $(d+2)\times(d+2)$ matrix $Y$ whose $(i,j)^{th}$ entry is $Y_{i,j} = \alpha_i\alpha_jf(y+ih + j(h_1 + ih_2))$ except when $j=0$. When $j=0$, $Y_{i,0} = \alpha_i \alpha_0 g(y+ih)$.
\[ Y = \begin{bmatrix}
   \alpha_0\alpha_0 g(y) &  \dots  & \alpha_0\alpha_j f(y+ jh_1) & \ldots \\
    \vdots &  \ddots &  \vdots  & \ddots  \\
    \alpha_i \alpha_0g(y +ih)  & \ldots & \alpha_i\alpha_jf(y+ih + j(h_1 + ih_2)) & \ldots \\
    \vdots  & \ddots  & \vdots & \ddots
\end{bmatrix}
\]
 Define the following set of events:
 \begin{itemize}
 \item For $i\in \{0,1,\ldots, d+1\}$, $R_i$ be the event that the sum of all elements from row $i$ is $zero$, i.e $\sum_{i=0}^{d+1} Y_{i,j} = 0$.
 \item For $j\in \{0,1,\ldots, d+1\}$, $C_j$ be the event that the sum of all elements from column $j$ is $zero$, i.e  $ \sum_{j=0}^{d+1} Y_{i,j} = 0$.
 \end{itemize}
Let $h_1, h_2$ are independent and distributed u.a.r in $\F^m$.  As the event $C_0$ is independent of $h_1$ and $h_2$, in order to prove the claim it is enough to show that $\Pr_{h_1, h_2} [C_0] > 0$.

For each row $i\in \{0, 1, 2, \ldots, d+1\}$ we apply \lref[Claim]{claim:agreement at y} with $y' = y+ih$ and $h' = h_1+ih_2 $, and get $\Pr_{h_1, h_2} [\neg R_i] \leq 2(d+1)\delta$  (note that $\alpha_0 = -1)$. If $h_1, h_2$ are independent and distributed u.a.r in $\F^m$ then so are $(y+jh_1)$ and $(h+h_2)$. Therefore,  using the hypothesis ~(\ref{eq:locally low deg}) of the theorem and \lref[Fact]{fact:low deg char}, we have
for all columns except $j=0$, $\Pr_{h_1, h_2}[\neg C_j] \leq \delta$.  Using union bound, we get
$$\Pr_{h_1, h_2}\left[\mathop{\wedge}_{i=0}^{d+1} R_i \mathop{\wedge}_{j=1}^{d+1} C_j\right] \geq 1-2(d+1)(d+2)\delta + (d+1)\delta > 0.$$
The claim now follows using the observation that the event $C_0$ is implied by the event $\mathop{\wedge}_{i=0}^{d+1} R_i \mathop{\wedge}_{j=1}^{d+1} C_j$. To see this, the event $\mathop{\wedge}_{i=0}^{d+1} R_i$ implies that the sum of all entries in $Y$ is $zero$ whereas $\mathop{\wedge}_{j=1}^{d+1} C_j$ implies that the sum of all elements from the submatrix $(Y_{i,j})_{j=1}^{d+1}$ is $zero$. Hence, if both these events happen then the sum of all elements from column $0$ must be $zero$.
%
%
%
\end{proof}

\end{document}